\documentclass[journal]{settings/IEEEtran}
\IEEEoverridecommandlockouts
\usepackage[english]{babel}
\makeatletter
\renewcommand{\fnum@figure}{\textbf{\figurename~\thefigure}}
\makeatother

\usepackage{mathrsfs}
\usepackage{type1cm} % type1 computer modern font
\usepackage{graphicx} % advanced figures
\usepackage{xspace} % fix space in macros
\usepackage{balance} % to better equalize the last page
\usepackage{booktabs} % nicer tables
\usepackage{multirow} % multi rows for tables
\usepackage{subcaption} % subfloats
\usepackage{bold-extra} % bold + {small capital, italic}
\usepackage{algorithm}
\usepackage{algorithmic}
\usepackage{microtype} % compress text
\usepackage{siunitx} % \num for decimal grouping
\usepackage{xfrac} % nicer slanted fractions
\usepackage{amsmath}
\usepackage{amssymb}
\usepackage{mathtools}
\usepackage{amsfonts}
\usepackage{amsthm}
\usepackage[hyphens]{url} % handle long urls
\usepackage[bookmarks, pdftex, colorlinks=true]{hyperref} % clickable references
\usepackage{cleveref} % smart references
\usepackage{csquotes}
\usepackage{bm}
\usepackage{xspace}
\usepackage{newunicodechar}
\newunicodechar{⟓}{\ensuremath{\uplus}}
\usepackage{tikz}
\usetikzlibrary{positioning}
\usetikzlibrary{bayesnet}
\usetikzlibrary{arrows}
\usetikzlibrary{shapes}
\usetikzlibrary{fit}
\usepackage{tikz-cd}
\usetikzlibrary{calc, positioning, arrows.meta, shapes.geometric, fit, backgrounds}
\usepackage{pgf}
\usepackage{tikz-3dplot}
\usepackage{tcolorbox}
\usepackage{pgfplots}
\pgfplotsset{compat=1.18}
\usepackage[show]{chato-notes}
\usepackage{cite}

\newcommand{\spara}[1]{\smallskip\noindent\textbf{#1}}

\newcounter{squishenumcounter}
\newenvironment{squishenum}
{\begin{list}{\arabic{squishenumcounter}.}
    {\usecounter{squishenumcounter}
     \setlength{\itemsep}{2pt}
     \setlength{\parsep}{1pt}
     \setlength{\topsep}{2pt}
     \setlength{\partopsep}{0pt}
     \setlength{\leftmargin}{2em}
     \setlength{\labelwidth}{1em}
     \setlength{\labelsep}{0.5em} } }
{\end{list}}

\usepackage[dvipsnames]{xcolor}
\hypersetup{
   bookmarks, pdftex,
   colorlinks=true,
   pagebackref=true, backref=page,
   linkcolor={red!50!black},
   filecolor={green!50!black},
   citecolor={green!50!black}, 
   urlcolor={blue!80!black},
}

\crefname{theorem}{Thm.}{Thms.}
\crefname{proposition}{Prop.}{Props.}
\crefname{lemma}{lem.}{lems.}
\crefname{corollary}{Cor.}{Cors.}
\crefname{definition}{Def.}{Defs.}
\crefname{section}{Sec.}{Secs.}
\crefname{figure}{Fig.}{Figs.}
\crefname{problem}{Prob.}{Probs.}
\crefname{appendix}{App.}{Apps.}
\crefname{equation}{Eq.}{Eqs.}
\crefname{algorithm}{Alg.}{Algs.}

\graphicspath{{fig}}

\theoremstyle{plain}
\newtheorem{theorem}{Theorem}

\newtheorem{lemma}{Lemma}

\newtheorem{definition}{Definition}

\newtheorem{remark}{Remark}

\DeclareMathOperator*{\argmin}{arg\,min}

\newcommand{\reall}{\ensuremath{\mathbb{R}}\xspace}

\newcommand{\edgeset}{\ensuremath{\mathcal{E}}\xspace}

\newcommand{\zeros}{\ensuremath{\boldsymbol{0}}\xspace}

\newcommand{\eye}[1]{\ensuremath{\mathbf{I}_{#1}}\xspace}
\newcommand{\stiefel}[2]{\ensuremath{\mathrm{St}({#1},{#2})}\xspace}
\newcommand{\ort}[1]{\ensuremath{\mathrm{O}({#1})}\xspace}

\newcommand{\vertexset}{\ensuremath{\mathcal{V}}\xspace}

\newcommand{\w}{\ensuremath{\mathbf{w}}\xspace}
\newcommand{\e}{\ensuremath{\boldsymbol{\varepsilon}}\xspace}
\newcommand{\x}{\ensuremath{\mathbf{x}}\xspace}
\newcommand{\y}{\ensuremath{\mathbf{y}}\xspace}
\newcommand{\X}{\ensuremath{\mathbf{X}}\xspace}
\newcommand{\z}{\ensuremath{\mathbf{z}}\xspace}
\newcommand{\Lapl}{\ensuremath{\mathbf{L}}\xspace}

\newcommand{\U}{\ensuremath{\mathbf{U}}\xspace}
\newcommand{\V}{\ensuremath{\mathbf{V}}\xspace}
\newcommand{\Z}{\ensuremath{\mathbf{Z}}\xspace}
\newcommand{\Omatrix}{\ensuremath{\mathbf{O}}\xspace}

\newcommand{\myP}{\ensuremath{\mathbb{P}}\xspace}
\newcommand{\myB}{\ensuremath{\mathbb{B}}\xspace}

\definecolor{mypurple}{RGB}{254, 68, 218}

\newcommand{\connection}{\ensuremath{\mathbb{G}}\xspace}
\newcommand{\Kronecker}{\ensuremath{\mathbb{K}}\xspace}

\newcommand{\connectionL}{\ensuremath{\mathbb{L}}\xspace}
\newcommand{\bdO}{\ensuremath{\mathbb{O}}\xspace}

\newcommand{\graph}{\ensuremath{G}\xspace}

\newcommand{\specialO}[1]{\ensuremath{\mathrm{SO}(#1)}\xspace}
\newcommand{\specialOprod}[2]{\ensuremath{\mathrm{SO}(#1)^{#2}}\xspace}

\makeatletter
\renewenvironment{proof}[1][\proofname]{%
  \par
  \pushQED{\qed}%
  \normalfont
  \topsep6\p@\@plus6\p@\relax
  \trivlist
  \item[\hskip\labelsep\itshape#1\@addpunct{.}]%
  \ignorespaces
}{%
  \popQED\endtrivlist\@endpefalse
}
\makeatother

\title{Structured Sheaf Learning of Consistent \\ Connection Graphs}

\begin{document}
\author{Leonardo Di Nino, ~\IEEEmembership{Student Member,~IEEE,} Gabriele D'Acunto, ~\IEEEmembership{Member,~IEEE,}\smallskip\\ Sergio Barbarossa, ~\IEEEmembership{Life Fellow,~IEEE} and Paolo Di~Lorenzo,~\IEEEmembership{Senior Member,~IEEE}
        \vspace{-.5cm}
        % <-this % stops a space
\thanks{The authors are with the Dept. of Information Engineering, Electronics, and Telecommunications, Sapienza University of Rome, Via Eudossiana, 18, 00185 Rome, Italy. Email: \{leonardo.dinino, gabriele.dacunto, sergio.barbarossa paolo.dilorenzo\}@uniroma1.it. The work was supported by the SNS JU project 6G-GOALS \cite{strinati2024goal} (Horizon 2020 Grant no. 101139232), and by Huawei Technology France SASU under Grant N. Tg20250616041.}}

\maketitle

\begin{abstract}
Connection graphs (CGs) extend classical graphs by associating vector-valued signals to nodes and orthogonal transport maps across edges, making them a natural model for synchronization and manifold-based signal processing. Despite their growing use, learning CGs directly from observations remains challenging because the network topology and the underlying geometric structure are coupled through non-Euclidean orthogonality constraints. In this work, we address this inverse problem by learning a consistent connection graph from noisy vector-valued signals. Exploiting the spectral characterization of consistent CGs, we formulate a structured learning problem that jointly estimates a denoised signal, the graph topology, and node-wise local reference frames. The proposed formulation couples the spectrum of the learned connection Laplacian to that of an underlying combinatorial Laplacian, enabling explicit spectral and topological priors while guaranteeing a nontrivial global-section space. We develop Structured Connection Graph Learning (SCGL), a block-coordinate algorithm that combines closed-form updates, manifold projections, and spectral constraints, and converges to stationary points of the resulting nonconvex problem. Numerical experiments show that SCGL improves topology and geometry recovery over competing approaches, while also yielding effective denoising and signal-compression bases.
\end{abstract}

\begin{keywords}
Graph Learning, graph signal processing, sheaf signal processing, connection graphs. \vspace{-.2cm}
\end{keywords}

\section{Introduction}\label{sec:introduction}

In recent years, Graph Signal Processing (GSP)~\cite{sandryhaila2013discrete} has emerged as a powerful framework for analyzing data supported on irregular domains. 
Its cornerstone is the graph shift operator, typically the graph Laplacian~\cite{smola2003kernels}: under the assumption that similarities are encoded in local neighborhoods, this operator provides a surrogate for distance in domains where conventional metrics are unavailable or uninformative.
Such domain-aware representations stem from the algebraic and spectral properties of the shift operator, which enables diffusion processes on graphs and underpin both classical convolutional methods~\cite{leus2023graph} and graph-based deep learning architectures~\cite{wu2020comprehensive}.
Consequently, learning graphs and their associated shift operators from data is a fundamental problem in GSP \cite{segarra2017network,Sardellitti2019GraphTopology,mateos2019connecting}.

Graph learning is, in general, an ill-posed combinatorial problem and is therefore guided by prior knowledge of properties that the observed signals are expected to exhibit on the inferred system. 
A wide range of methods, spanning signal processing and probabilistic modeling~\cite{mateos2019connecting} as well as deep learning~\cite{Kazi_2023}, have been proposed.
Broadly, they follow two main strategies:
\emph{(i)} enforcing a desired topology by balancing signal sparsity and connectivity;
and \emph{(ii)} promoting signal smoothness, i.e., low variation across connected nodes on the learned graph. The smoothness assumption is widespread in graph learning because it is intimately tied to the behavior of diffusion processes and to the equilibrium states encoded in the spectra of shift operators. 
Conversely, the over-smoothing phenomenon~\cite{rusch2023survey,keriven2022not}, which arises as a consequence of such dynamics, can degrade performance in tasks that require preserving fine-grained details such as long-range interactions or high-frequency patterns---revealing the limits of smoothness when applied beyond its appropriate context.

Beyond graph models, there has been growing interest in higher-order structures that capture dependencies which classical graphs struggle to represent.
Topological signal processing and deep learning have emerged as alternatives to GSP by encoding multi-way dependencies in higher-order combinatorial topologies~\cite{Barbarossa2020Topological,isufi2025topological,pmlr-v235-papamarkou24a}, 
%yet they still cannot capture processes that involve diffusion of heterogeneous structured data. 
yet they offer limited flexibility for modeling and processing heterogeneous, structured data.
In this direction, sheaf-theoretic approaches~\cite{curry2014sheaves,ayzenberg2025sheaf} enrich graph-based representations by assigning structured data to nodes and edges and relating them via structure-preserving maps, yielding a general \emph{functorial} framework that can be specialized to the underlying process or representation task.
By defining local-to-global and global-to-local dynamics across a network, sheaves move beyond standard graph diffusion~\cite{bodnar2023neuralsheafdiffusiontopological} by intrinsically shaping the sheaf Laplacian~\cite{hansen2020laplacians}, which governs processes with richer representational expressivity.
This additional degree of freedom comes at the cost of a more complex spectral behavior~\cite{hansen2019toward}, which now depends not only on the underlying topology but also on the structure-preserving maps.
This shift highlights the need for new coupling principles between signals and systems that ensure a meaningful representation framework and, ideally, enforce desirable spectral properties. 
As a result, learning sheaf Laplacians from data has emerged as a fundamental and largely open challenge~\cite{8683709,10942997}.

Among sheaf-theoretic models, \emph{connection graphs} (CGs) consist of \emph{(i)} nodes and edges associated with inner-product vector spaces and \emph{(ii)} structure-preserving maps given by orthogonal linear transformations.
CGs have gained prominence in manifold learning~\cite{singer2012vector}, synchronization problems~\cite{bandeira2013cheeger}, shape analysis~\cite{sharp2019vector}, Riemannian signal processing~\cite{battiloro2024tangent}, and neural sheaf diffusion~\cite{barbero2022sheaf}. 
Motivated by these applications, we investigate the problem of learning CGs from observed signals.

\spara{Related works.} 
CGs find ubiquitous applications in signal processing and machine learning, supported by solid theoretical contributions that define them as a tool for processing vector-valued data over discrete topologies.
Vector Diffusion Maps (VDM) provide a seminal framework that approximates connection Laplacians through a nonlinear dimensionality reduction method grounded in geometric principles~\cite{singer2012vector}. 
Under suitable assumptions, as the number of samples from the underlying manifold tends to infinity, VDM converges to the Laplace--Beltrami operator on the manifold's tangent bundle. 
The link between the Laplace--Beltrami operator and its discretized counterpart underlies tangent bundle convolutional learning ~\cite{battiloro2024tangent}, which generalizes scalar-field processing over discrete manifolds~\cite{wang2022convolutional} to vector-field processing. 
This framework builds on vector heat diffusion~\cite{sharp2019vector}, parametrized by the connection Laplacian, to implement finite impulse response filters on vector-valued signals defined over the tangent bundle of a manifold. 
By concatenating such filters with nonlinearities, the method yields principled sheaf neural networks, obtained through a discretization of both the underlying continuous process and the geometry.
Similarly, the authors of~\cite{barbero2022sheaf} discretize heat sheaf diffusion to build convolutional neural networks parametrized by a connection Laplacian learned in a preprocessing step, which acts as a form of graph alignment. 
This choice significantly reduces the complexity compared to more general sheaf neural networks~\cite{hansen2020sheaf,barbero2022sheafAttention}, where the Laplacian is learned at runtime.

Despite the widespread adoption of these models, contributions toward a solid spectral theory that mirrors convolutional learning remain relatively sparse. 
While~\cite{hansen2019toward} outlines the program for a spectral theory of cellular sheaves,~\cite{bandeira2013cheeger} provides a tighter result tailored to connection graphs in the form of a Cheeger-like inequality, shown to act as a lower bound for spectral methods that solve the synchronization problem of recovering orthonormal transformations from noisy measurements. 
In~\cite{chung2014ranking}, the authors derive necessary and sufficient conditions for the spectrum of the connection Laplacian to mirror that of the underlying graph, up to a scaling determined by the dimension of the vector spaces over nodes.
This property---known as \emph{consistency}---underpins recent neural architectures~\cite{bamberger2024bundle} that have shown remarkable performance in overcoming inherent limitations of graph machine learning.

The consistency property plays a central role in our framework: we address the problem of learning the connection Laplacian from observed signals and therefore require criteria that bias the learning toward desired topological and geometric structures, ideally allowing direct control over the spectrum of the inferred operator.
In this sense, our work generalizes structured graph learning as defined in~\cite{kumar2020unified}, where the authors propose an algorithm that encodes prior topological assumptions into the learning procedure via spectral constraints while preserving the smoothness assumption.
Although the smooth graph learning paradigm~\cite{dong2016learning} has already been extended to sheaf Laplacians in~\cite{8683709} via semidefinite programming, connection graphs call for ad-hoc learning criteria to recover proper operators.
Indeed, the semidefinite programming method in~\cite{8683709} can recover certain classes of sheaves but not connection graphs, which are inherently tied to the non-Euclidean geometry of the orthogonal manifold and thus cannot be properly retrieved via convex programming. In~\cite{10942997}, we addressed the limits of this framework using Procrustes alignment and binary edge sampling, further motivating the need for ad-hoc criteria in sheaf representation learning. However, a structured learning framework that jointly infers the topology and orthogonal transports of a consistent connection graph from signals, while explicitly controlling the spectrum of the resulting operator, is still missing. This gap motivates the approach developed next.

\spara{Contributions.} In this work, we propose a framework to learning consistent connection Laplacians from noisy vector-valued observations by parameterizing the operator leveraging consistency hypoteses~\cite{chung2014ranking} through a combinatorial Laplacian and node-wise orthogonal frames, jointly recovering topology and local alignments. 
Consistency forces the learned spectrum to match the combinatorial one up to multiplicity and its kernel to contain synchronized global sections, so that well-understood spectral priors on graph Laplacians~\cite{chung1997spectral,kumar2020unified} directly constrain the inferred connection graph.

Our main contributions can be summarized as follows:
\begin{squishenum}
\item[\emph{(i)}] We formulate the joint inference of the connection Laplacian and filtered signals from noisy observations, extending spectral control from graphs to connection graphs via consistency and biasing learning toward desired structures while guaranteeing targeted geometric properties.
\item[\emph{(ii)}] We devise the \textbf{S}tructured \textbf{C}onnection \textbf{G}raph \textbf{L}earning (SCGL) algorithm, which recovers topological and geometric patterns jointly via block descent over the relevant manifolds, with provable convergence.
\item[\emph{(iii)}] We assess SCGL extensively on synthetic and real-world data, showing favorable performance against baselines across a range of graph generation models.
\end{squishenum}
This work extends our preliminary paper in \cite{di2026learning} both methodologically and empirically, with a revised formulation, optimization framework, and broader experimental validation.

\spara{Roadmap.} 
\cref{sec:bg} reviews the fundamentals of cellular sheaf theory, focusing on global sections and their existence in connection graphs.
\cref{sec:formulation} defines the signal model and derives the learning problem from it.
Then, \cref{sec:solution} outlines the algorithmic routines and initialization.
Next, \cref{sec:results} provides numerical results validating the proposed methodology. 
Finally, \cref{sec:conclusion} concludes and outlines future work. 

\spara{Notation.} 
Scalars, vectors, matrices, and block matrices are denoted by plain letters $s$, bold lowercase letters $\mathbf{v}$, bold uppercase letters $\mathbf{M}$, and blackboard bold uppercase letters $\mathbb{B}$, respectively. 
The $i$-th row of $\mathbf{M}$ is denoted by $[\mathbf{M}]_{i,:}$, the $j$-th column by $[\mathbf{M}]_{:,j}$, and the $(i,j)$-th entry by $[\mathbf{M}]_{i,j}$. 
Similarly, for a block matrix $\mathbb{B}$, the symbol $[\mathbb{B}]_{i,:}$ denotes the $i$-th block-row, $[\mathbb{B}]_{:,j}$ the $j$-th block-column, $[\mathbb{B}]_{i,j}$ the $(i,j)$-th block.
The direct sum of vector spaces is denoted by $\bigoplus$, $\mathrm{ker}(\cdot)$ denotes the kernel of a linear operator.
The determinant of $\mathbf{M}$ is $\mathrm{det}(\mathbf{M})$, while $\mathrm{gdet}(\mathbf{M})$ denotes the generalized determinant.
The Kronecker product is denoted by $\otimes$. 
$\mathbb{E}[\cdot]$ and $\mathbb{V}[\cdot]$ denote expectation and variance of random quantities.
The set of orthogonal matrices of size $n$ is $\ort{n}\coloneqq\{\Omatrix \in \reall^{n\times n} \mid \Omatrix^\top=\Omatrix^{-1}\}$; the set of special orthogonal matrices is $\mathrm{SO}(n)\coloneqq \{\Omatrix \in \ort{n} \mid \det(\Omatrix)=1\}$; and the Stiefel manifold is 
$\stiefel{n}{k}\coloneqq\{\mathbf{V} \in \reall^{n\times k} \mid \V^\top \V=\eye{k} \}$
with $k<n$.
\section{Background}\label{sec:bg}

This section introduces network sheaves, connection graphs, and the algebraic tools used throughout the paper.

\subsection{A Primer on Network Sheaves}

A network sheaf enriches a graph with algebraic data describing how signals are structured over nodes and edges and how this structure is preserved as information propagates across the network. Its formal definitions reads as follows  \cite{curry2014sheaves}:
\begin{definition}[Network sheaf in $\mathsf{Vect}$]
Let $\mathcal{G}(\mathcal{V},\mathcal{E})$ be a graph equipped with the incidence relation $\trianglelefteq_\mathcal{G}:\mathcal{V}\rightarrow\mathcal{E}$. A network sheaf $\mathcal{F}:(\mathcal{V}\cup\mathcal{E},\trianglelefteq_\mathcal{G})\rightarrow \mathsf{Vect}$ is specified by:
\begin{enumerate}
    \item a vector space $\mathcal{F}(v)$ for each node $v \in \mathcal{V}$, called a \textit{node stalk};
    \item a vector space $\mathcal{F}(e)$ for each edge $e \in \mathcal{E}$, called an \textit{edge stalk};
    \item a linear map $\mathcal{F}(v\trianglelefteq e)$ for each incidence relation $v\trianglelefteq e$, called a \textit{restriction map}.
\end{enumerate}
\end{definition}
Informally, a network sheaf equips each node with a local space in which its signal is represented and each edge with a space in which neighboring node signals can be compared or combined. The restriction maps specify how a node-level signal is translated into the representation associated with an incident edge. Hence, rather than assuming that all node signals share the same representation and can be directly compared, a sheaf explicitly models the local transformations required to relate information across the network. A collection of node signals is globally compatible when their edge-level representations agree on every edge.

\begin{remark}
Assigning $\mathcal{F}(v) = \mathbb{R}$ for all $v \in \mathcal{V}$, $\mathcal{F}(e) = \mathbb{R}$ for all $e \in \mathcal{E}$, and $\mathcal{F}(v\trianglelefteq e) = 1$ for all $v\trianglelefteq e$ recovers the unweighted graph used in graph signal processing. A network sheaf therefore subsumes this setting as a special case.
\end{remark}

Defining the space of \emph{$0$-cochains} $\mathcal{C}^0 =\bigoplus_{v \in \mathcal{V}} \mathcal{F}(v)$ and the space of \emph{$1$-cochains} $\mathcal{C}^1 =\bigoplus_{e \in \mathcal{E}} \mathcal{F}(e)$, we can introduce the coboundary operator $\delta: \mathcal{C}^0\rightarrow\mathcal{C}^1$. 
The coboundary operator $\delta$ encodes the local-to-global transition induced by a sheaf, describing how signals observed at nodes coherently glue over the interconnecting edges: it reads edge wise as
\begin{equation}
    (\delta\mathbf{x})_e = \mathcal{F}(v \trianglelefteq e)\,\mathbf{x}_v - \mathcal{F}(u \trianglelefteq e)\,\mathbf{x}_u .
\end{equation}
% \begin{definition}[Coboundary map]
% Let $\mathcal{G}(\mathcal{V},\mathcal{E})$ be a graph and $\mathcal{F}:(\mathcal{V}\cup\mathcal{E},\trianglelefteq_{\mathcal{G}})\rightarrow \mathsf{Vect}$ a network sheaf. The coboundary map is the operator $\delta:\mathcal{C}^0\rightarrow\mathcal{C}^1$ acting on $\mathbf{x} \in \mathcal{C}^0$ as

% \end{definition}

The coboundary operator naturally induces a Laplacian on node signals by composition with its adjoint. The resulting sheaf Laplacian generalizes the standard graph Laplacian to signals and relations structured by the sheaf~\cite{hansen2020laplacians}.
\begin{definition}[Sheaf Laplacian]\label{def:Sheaf_laplacian}
Let $\mathcal{G}(\mathcal{V},\mathcal{E})$ be a graph and $\mathcal{F}:(\mathcal{V}\cup\mathcal{E},\trianglelefteq_{\mathcal{G}})\rightarrow \mathsf{Vect}$ a network sheaf. The sheaf Laplacian is the operator $\mathbb{L}_\mathcal{F}:\mathcal{C}^0\rightarrow\mathcal{C}^0$ defined by
\begin{equation}
    \mathbb{L}_\mathcal{F}=\delta^\top\delta ,
\end{equation}
whose blocks satisfy
\begin{align}
    [\mathbb{L}_\mathcal{F}]_{u,u} &= \sum_{u \trianglelefteq e} \mathcal{F}(u \trianglelefteq e)^\top \mathcal{F}(u \trianglelefteq e), \label{eq:Laplacian_block_1}\\
    [\mathbb{L}_\mathcal{F}]_{u,v} &= - \mathcal{F}(u\trianglelefteq e)^\top \mathcal{F}(v\trianglelefteq e). \label{eq:Laplacian_block_2}
\end{align}
\end{definition}
A nontrivial $0$-cochain that achieves global alignment, i.e., $\mathcal{F}(v \trianglelefteq e)\,\mathbf{x}_v = \mathcal{F}(u \trianglelefteq e)\,\mathbf{x}_u$ for all $e \in \mathcal{E}$, is called a \emph{global section}. The space of global sections, denoted by $\mathcal{H}^0\subseteq\mathcal{C}^0$, coincides with the null space of both the coboundary operator and the sheaf Laplacian: $\mathcal{H}^0=\ker(\delta)=\ker(\mathbb{L}_{\mathcal{F}})$. Global sections spaces underpin sheaf-based models of distributed coordination and decision-making~\cite{hansen2019distributed,hansen2021opinion,issaidtackling,ghalkha2026sheafalign,grimaldi2026sheaftheoreticframeworkdistributedmultisite}. Viceversa, inferring a sheaf with guaranteed global sections from data is more challenging than learning a standard graph: while graph topology directly determines key spectral properties with connectivity guaranteeing constant global sections, in a general sheaf the restriction maps also shape the spectrum and may prevent nontrivial global sections even on a connected graph. CGs provide a structured class of sheaves in which the edge maps are orthogonal and, under consistency, the kernel of the associated Laplacian is explicitly characterized and directly linked to the spectrum of the underlying combinatorial graph.

\subsection{Connection Graphs}

A connection graph is a particular network sheaf in which all node and edge stalks are copies of $\mathbb{R}^n$ and the restriction maps associated with an edge $e=(u,v)$ are orthogonal. Since the edge stalk provides a common reference space, one of the two maps can be fixed to the identity without loss of generality; the remaining map is then represented by a single orthogonal transformation $\Omatrix_{uv}$, which aligns the signal coordinates at $v$ with those at $u$. The formal definition is given as follows.

\begin{definition}[Connection graph~\cite{chung2014ranking}]\label{def:connection}
A connection graph is a triple $\connection = \langle G, \reall^n, \ort{n} \rangle$ consisting of
(i) an undirected graph $\graph\coloneqq(\vertexset,\edgeset,\w)$ with vertex set $\vertexset$, edge set $\edgeset$, and edge weights $\w_e>0$ for each $e \in \edgeset$;
(ii) an $n$-dimensional inner-product space $\reall^n$ attached to each node $v \in \vertexset$;
and (iii) a set of orthogonal matrices $\Omatrix_e \in  \ort{n}$, one for each $e \in \edgeset$.
\end{definition}

The algebraic operators used to describe graphs extend naturally to
connection graphs, yielding a single object that encodes both combinatorial
adjacency and geometric structure generalizing the combinatorial graph Laplacian.

\begin{definition}[Connection Laplacian]\label{def:connectionL}
Let $\connection$ be a connection graph with $|\vertexset| = V$ nodes. The connection Laplacian $\connectionL \in \reall^{Vn \times Vn}$ is the block matrix with $[\connectionL]_{i,j}=\zeros$ for each $(i,j) \notin \edgeset$ and, for each edge $(i,j) \in \edgeset$,
\begin{equation}\label{eq:KLapl}
        [\connectionL]_{i,j}  \coloneqq \begin{cases}
            -\w_{ij} \Omatrix_{ij} & i>j,\\[2pt]
            [\connectionL]_{j,i}^\top & i<j,\\[2pt]
            \sum_{j\neq i} \w_{ij}\eye{n} & i=j.
        \end{cases}
    \end{equation}
\end{definition}
The connection Laplacian in~\eqref{eq:KLapl} is a particular instance of the sheaf Laplacian introduced previously in Definition \ref{def:Sheaf_laplacian}.

A global section over a CG requires
\begin{equation}\label{eq:sync}
    \mathbf{x}_i = \mathbf{O}_{ij}\mathbf{x}_j, \quad \forall\,(i,j)\in\mathcal{E},
\end{equation}
which can be understood as a form of \emph{synchronization}~\cite{bandeira2013cheeger} over the orthogonal group: graph consensus modulated by rotation matrices acting as alignment operators (cf. Fig. \ref{fig:CG_consistency}). When $\mathbf{O}_{ij}=\mathbf{I}_n$ for all $(i,j)\in\mathcal{E}$, this reduces to the \emph{trivial bundle}~\cite{steenrod1999topology}, whose global sections coincide with standard graph consensus. CGs therefore model a strictly richer family of equilibrium states than classical graph constructions.

In CGs, the existence of global sections is tied to the spectral structure of the connection Laplacian. Unlike the graph case, where connectivity alone determines the dimension of the Laplacian kernel, the topology of a CG is insufficient to fully characterize its spectral properties or its global-section space. A more informative characterization is obtained through the notion of \emph{consistency}: the orthogonal maps on the edges compose to the identity around every cycle.

\begin{theorem}[Spectral characterization~\cite{chung2014ranking}]\label{th:consistency_G}
Let $\connection$ be a connection graph with $|\vertexset|=V$ nodes and connection Laplacian $\connectionL \in \reall^{Vn \times Vn}$, and let $\mathbf{L}$ be the combinatorial Laplacian of the underlying graph $\graph$. The following are equivalent:
\begin{enumerate}
    \item \label{th:consistency_G_1} $\connection$ is consistent;
    \item \label{th:consistency_G_2} the eigenvalues of $\connectionL$ are those of $\mathbf{L}$, each with multiplicity $n$;
    \item\label{th:consistency_G_3} for each node $i$ in $\graph$ there exists $\Omatrix_i \in \mathrm{SO}(n)$ such that $\Omatrix_{ij} = \Omatrix_i^\top \Omatrix_j$ for every edge $(i,j) \in \edgeset$.
\end{enumerate}
\end{theorem}
\begin{figure}
    \centering
    \includegraphics[width=0.75\linewidth]{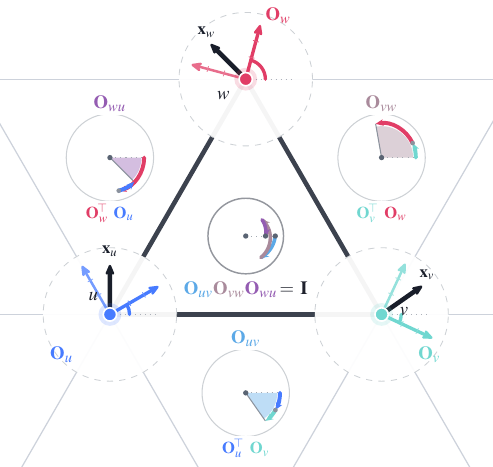}
    \caption{A zoom-in on a cycle of a consistent connection graph. Each node is equipped with a reference frame and restriction maps are structurally defined by their composition: by design, this yields composition to the identity around the cycle. Each node observes a signal for which agreement is reached only relatively to the reference frame.}
    \label{fig:CG_consistency}
\end{figure}
Statements~\ref{th:consistency_G_2} and~\ref{th:consistency_G_3} characterize consistent connection graphs. In particular, the spectrum of $\mathbb{L}$ consists of the eigenvalues of $\mathbf{L}$, each repeated $n$ times. Moreover, any connected graph can be lifted to a consistent CG by assigning orthogonal reference frames to its nodes and defining edge transports through their relative alignments. Under Theorem~\ref{th:consistency_G},
\begin{equation}\label{eq:laplacian_factiorization}
    \mathbb{L} = \bdO^\top(\mathbf{L}\otimes\mathbf{I}_n)\bdO,
\end{equation}
where $\bdO=\mathrm{blkdiag}(\{\mathbf{O}_v\}_{v\in\mathcal{V}})$ and $\mathbf{O}_v\in\mathrm{SO}(n)$. Hence, for a connected graph, $\ker(\mathbb{L})$ has dimension $n$ and contains synchronized signals, which become constant in a common reference frame.

\begin{remark}\label{remark:gauge_invariance}
The system possesses gauge invariance: because of the symmetry structure of the $\mathrm{SO}(n)$ group, the node rotations are identifiable only up to an arbitrary global rotation within each connected component.
\end{remark}

Equation~\eqref{eq:laplacian_factiorization} is a structural key property:
under consistency, a CG is entirely specified by the combinatorial topology
of $G$ and by node-wise reference frames $\{\mathbf{O}_v\}_{v \in \mathcal{V}}$,
whose global sections are precisely the synchronized signals.
This suggests
recovering a CG by jointly estimating topology and local frames from
observations assumed to be synchronizable. In practice, however, noise,
intrinsic variability, and model mismatch push the observations away from
$\ker(\mathbb{L})$, making a hard global-section constraint brittle. We
therefore require only concentration around it, and we model this principle through a probabilistic prior
supported on the low-frequency subspace of the connection Laplacian: the
synchronization requirement becomes a smoothness penalty, and the recovery of
$\mathbb{G}$ a statistical estimation problem, as developed next.

\section{Problem Formulation}\label{sec:formulation}
In this section, we first introduce a model for
vector-valued signals over a consistent CG, and then derive from it a joint
estimation problem in the denoised signals, the graph topology, and the node reference frames.

We consider a collection of possibly noisy measurements encoded in a data matrix $\X = [\x_1, \dots, \x_M]$, associated with an unknown consistent CG $\connection = \langle \graph, \reall^n, \ort{n} \rangle$, where $\graph$ has $V$ nodes. Each column $\x_m \in \reall^{Vn}$ aggregates the local signals $\x_i \in \reall^n$ observed at the nodes $i\in\vertexset$. 
We model such signals through a factor-analysis framework \cite{bartholomew2011latent}, which has been successfully used to represent smooth graph signals using a reduced set of latent variables\cite{dong2016learning}. 
In this
setting, smoothness with respect to $\connectionL$ promotes global sections as a soft prior rather than imposing them as a hard constraint.
Let $\connectionL = \V \bm{\Gamma} \V^\top$ be the eigendecomposition of the connection Laplacian. 
We introduce latent variables $\y \sim \mathcal{N}(\zeros_{Vn},\bm{\Gamma}^\dagger)$ and isotropic, independent Gaussian noise $\e \sim \mathcal{N}(\zeros_{Vn},\sigma^2\eye{Vn})$. Since the prior variance is larger along eigenvectors associated with smaller
eigenvalues, it favors latent signals concentrated on the low-frequency modes of $\connectionL$ and, consequently, quasi-harmonic observations. 
Assuming zero-centered observations, we model the observed signals as
$\x = \V \y + \e$;
therefore, the observed signal is marginally distributed as
\begin{equation}\label{eq:marginal-x}
    \x \sim \mathcal{N}(\mathbf{0}_{Vn}, \connectionL^\dagger + \sigma^2 \eye{Vn}) .
\end{equation}

Starting from~\eqref{eq:marginal-x} and following a maximum-a-posteriori approach for $\y$, as in~\cite{dong2016learning}, we apply the change of variables $\z=\V\y$. 
Treating the CG as unknown, this yields the learning problem below, which jointly estimates the denoised signals  $\Z = [\z_1, \dots, \z_M]$ and the connection Laplacian \connectionL from the noisy measurements $\X$:
\begin{equation}\label{eq:problem_raw}
\begin{aligned}
(\Z^*,\connectionL^*) = \argmin_{\Z,\; \connectionL\in\mathcal{CL}} \, & \;\;\frac{\gamma}{M}\|{\X-\Z}\|_{\mathrm{F}}^2 - \log \mathrm{gdet}(\connectionL)\\
& \;\;+ \mathrm{Tr}\{M^{-1}\Z^\top\Z\connectionL\} ,
\end{aligned}
\tag{P1}
\end{equation}
where $\gamma$ is a constant accounting for the noise variance, and $\mathcal{CL}$ represents the feasible set of Laplacians of consistent CGs, possibly endowed with desired topological properties.

Now, leveraging~\eqref{eq:laplacian_factiorization}, problem~\eqref{eq:problem_raw} can be recast in terms of the combinatorial Laplacian $\Lapl \in \mathcal{GL}$ and the block-diagonal basis $\bdO \in \specialOprod{n}{V}$ as:
\begin{equation}\label{eq:P2}
    \begin{aligned}
    \min_{\Z,\, \Lapl \in \mathcal{GL},\, \bdO \in \specialOprod{n}{V}} \; & 
    \frac{\gamma}{M}\|{\X-\Z}\|_{\mathrm{F}}^2 - \log\mathrm{gdet}\{\bdO^\top(\Lapl \otimes \eye{n})\bdO\} \\
    & + \mathrm{Tr}\{M^{-1}\Z\Z^\top\bdO^\top(\Lapl \otimes \eye{n})\bdO\}, 
    \end{aligned}
    \tag{P2}
\end{equation}
where $\mathcal{GL}$ is the set of valid combinatorial Laplacians. 
\begin{remark}
Although problem~\eqref{eq:problem_raw} is derived from a probabilistic signal model, it can equivalently be read as a classical minimum total-variation formulation, valid in any smooth-learning setting. 
Indeed, smooth graph learning does not require probabilistic modeling: it stems from a low-pass filtering assumption that encompasses~\eqref{eq:marginal-x} as a special case~\cite{kalofolias2016learn}. 
Under this view, problem~\eqref{eq:problem_raw} reduces to a minimum total-variation problem regularized by a log-barrier term controlling the dimension of the Laplacian kernel, which is in turn particularly well-suited to enforcing the spectral priors developed next, as it directly acts on the  structure of the operator of interest.
This interpretation carries over to the consistent connection graph setting in~\eqref{eq:P2}: the total-variation term retains its quadratic-form structure in the Laplacian, and can be equivalently expressed coordinate-wise over the rotated signals $\tilde{\mathbf{Z}}=\mathbb{O}\mathbf{Z}$ as:
\begin{align}
    & \mathrm{Tr}\{M^{-1}\mathbf{Z}\mathbf{Z}^\top\mathbb{O}^\top(\mathbf{L}\otimes\mathbf{I}_n)\mathbb{O}\}
    = M^{-1}\mathrm{Tr}\{\mathbb{O}\mathbf{Z}\mathbf{Z}^\top\mathbb{O}^\top(\mathbf{L}\otimes\mathbf{I}_n)\} \nonumber\\
    &= M^{-1}\mathrm{Tr}\{\tilde{\mathbf{Z}}\tilde{\mathbf{Z}}^\top(\mathbf{L}\otimes\mathbf{I}_n)\} = \sum_{l=1}^{n} M^{-1}\mathrm{Tr}\{\tilde{\mathbf{Z}}_l\tilde{\mathbf{Z}}_l^\top\mathbf{L}\} \label{eq:tv_coordwise}
\end{align}
where $\tilde{\mathbf{Z}}_l$ denotes the denoised signals expressed in the local reference frames and restricted to the $l$-th coordinate at each node.
Similarly, the logarithmic regularization is still encouraging
a non-degenerate spectrum for the Laplacian as it is invariant
under the isometric action of the orthogonal reference frames.
\end{remark}
To simplify the optimization, following~\cite{kumar2020unified}, we introduce a linear operator $\mathcal{L}_\Kronecker$ that maps the edge-weight vector $\w$ to the Kronecker-structured Laplacian $\Lapl\otimes\eye{n}$. This parametrization allows us to optimize directly over $\w$, rather than over $\Lapl\in\mathcal{GL}$.

\begin{definition}[Laplacian operator for the Kronecker graph]\label{def:KLapl}
The linear operator $\mathcal{L}_\Kronecker:\reall^{V(V-1)/2} \rightarrow \reall^{Vn \times Vn}$ is defined block-wise as
\begin{equation}\label{eq:KLaplOp}
[\mathcal{L}_\Kronecker(\w)]_{ij} \coloneqq
\begin{cases}
    -\w_{i+d_j}\, \eye{n} & i>j ,\\[2pt]
    [\mathcal{L}_\Kronecker(\w)]_{ji} & i<j ,\\[2pt]
    -\sum_{j\neq i} [\mathcal{L}_\Kronecker(\w)]_{ij} & i=j ,
\end{cases}
\end{equation}
where $d_j = -j + \tfrac{j-1}{2}(2V-j)$.
\end{definition}

Using~\cref{def:KLapl}, the equivalence in~\eqref{eq:laplacian_factiorization} can be recast as
\begin{equation}\label{eq:consistency_1}
    \connectionL = \bdO^\top(\Lapl \otimes \eye{n})\bdO= \bdO^\top \mathcal{L}_\Kronecker(\w)\,\bdO,
\end{equation}
which can be directly plugged into (\ref{eq:P2}). 

Beyond enforcing consistency, we wish to formulate a learning problem that admits topological priors on the CG. Since the CG topology is tightly coupled to the spectral properties of $\connectionL$, it is natural to introduce spectral priors as constraints~\cite{kumar2020unified}. These are conveniently imposed on the eigenvalues $\bm\Lambda$ of the combinatorial Laplacian, whose spectrum fully determines that of the connection Laplacian under consistency. 
Following~\cite{kumar2020unified}, we use the linear operator
$\mathcal{L}:\reall^{V(V-1)/2}\rightarrow\reall^{V\times V}$, obtained from the connection-Laplacian construction by setting $n=1$, which maps the edge-weight vector $\w$ to the combinatorial Laplacian:
\begin{equation}\label{eq:LaplLin}
[\mathcal{L}(\w)]_{ij} \coloneqq
\begin{cases}
    -\w_{i+d_j}, & i>j,\\[2pt]
    [\mathcal{L}(\w)]_{ji}, & i<j,\\[2pt]
    -\sum_{j\neq i}[\mathcal{L}(\w)]_{ij}, & i=j.
\end{cases}
\end{equation}
Its eigendecomposition 
\begin{equation}\label{eq:consistency_2_reload}
    \mathcal{L}(\w)=\U\bm\Lambda\U^\top,
\end{equation}
provides a non-redundant spectral parametrization on which to enforce the desired topological priors. Let $\mathcal{S}_{\bm\Lambda}$ denote the set of admissible spectra of connected graphs. Since every connected graph Laplacian has a single zero eigenvalue, enforcing $\bm\Lambda\in\mathcal{S}_{\bm\Lambda}$ fixes $\lambda_1=0$. Hence, we optimize only over the remaining $V-1$ eigenvalues and eigenvectors, transfering the log-barrier regularization directly on $\mathbf{\Lambda}$ being the variable of interest on which the regularization work; with a slight abuse of notation, we continue to denote them by $\bm\Lambda$ and $\U$, with $\U\in\stiefel{V}{V-1}$.

Finally, using \eqref{eq:consistency_1}, applying a Lagrangian relaxation of constraint~\eqref{eq:consistency_2_reload}, and adding an edge-sparsity regularizer $\Psi(\w)$, (\ref{eq:P2}) can be recast as the learning problem:
\begin{equation}\label{eq:final-prob}
    \begin{aligned}
        \min_{\Z,\, \bdO,\, \w,\, \U,\, \bm\Lambda}\;&
        \gamma M^{-1}\|\X - \Z\|_{\mathrm{F}}^2 - n \log\det(\bm\Lambda) \\
        & + \mathrm{Tr}\{ M^{-1}\Z\Z^\top \bdO^\top \mathcal{L}_\Kronecker(\w)\,\bdO\}
        + \alpha\, \Psi(\w) \\
        & + \frac{n\beta}{2}
        \|\mathcal{L}(\w) - \U\bm\Lambda\U^\top\|_{\mathrm{F}}^2 ,
    \end{aligned}
    \tag{P3}
\end{equation}
where $\Z \in \reall^{Vn \times M}$, $\bdO \in \specialOprod{n}{V}$, $\w \in \reall_{\geq 0}^{V(V-1)/2}$, $\U \in \stiefel{V}{V-1}$, and $\bm\Lambda \in \mathcal{S}_{\bm\Lambda}$; the strictly positive penalties $\alpha$ and $\beta$ are hyperparameters. Problem~\eqref{eq:final-prob} thus enables us to simultaneously denoise $\X$ and learn the connection Laplacian of a consistent CG with the desired topological properties.

\begin{remark}
In the noise-free setting of our preliminary work~\cite{di2026learning}, one simply sets $\Z = \X$, reducing the problem to maximum-likelihood estimation of the connection Laplacian under the same parametrization.
\end{remark}
\section{Structured Connection Graph Learning}\label{sec:solution}

The objective in~\eqref{eq:final-prob}, together with the Stiefel and product-manifold constraints, makes the problem nonconvex. We solve it efficiently and with guaranteed convergence to stationary points through a Riemannian block-coordinate scheme.

\subsection{Algorithmic Solution}\label{subsec:prob-solution}

We first decouple $\bdO$ from its manifold via the \emph{splitting of orthogonality constraints} (SOC) method~\cite{lai2014splitting}: we relax $\bdO \in \reall^{Vn \times Vn}$ and introduce a new variable $\myP \in \specialOprod{n}{V}$ such that
\begin{equation}\label{eq:P-constraint}
\myP - \bdO = \zeros_{Vn \times Vn} .
\end{equation}
As shown below, this yields closed-form updates for both $\bdO$ and $\myP$. Introducing constraint (\ref{eq:P-constraint}), Problem~\eqref{eq:final-prob} admits the (scaled) augmented Lagrangian:
\begin{equation}\label{eq:aul}
    \begin{aligned}
        \mathscr{L}(\Z, \bdO, \myP, \w, \U, \bm\Lambda) = &
        \;\frac{\gamma}{M} \|\X - \Z\|_{\mathrm{F}}^2 - n \log\det(\bm\Lambda) \\
        &\hspace{-2.5cm} + \mathrm{Tr}\{ M^{-1}\Z\Z^\top \bdO^\top \mathcal{L}_\Kronecker(\w)\bdO\}
        + \alpha\, \Psi(\w) \\
        &\hspace{-2.5cm} + \frac{n\beta}{2}
        \|\mathcal{L}(\w) - \U\bm\Lambda\U^\top\|_{\mathrm{F}}^2
        + \frac{\rho}{2}\|\bdO - \myP + \myB\|_{\mathrm{F}}^2 .
    \end{aligned}
\end{equation}
where $\rho>0$ is the penalty parameter, and $\myB \in \reall^{Vn\times Vn}$ denotes the scaled dual variable associated with (\ref{eq:P-constraint}). We look for saddle points of~\eqref{eq:aul} using the method of multipliers~\cite{boyd2011distributed}, alternating between minimization with respect to the primal blocks and maximization with respect to the dual variable. The resulting iterations consist of the closed-form block updates derived below, where $t$ denotes the iteration index.

\spara{Update for $\Z$.}
From \eqref{eq:aul}, the subproblem in $\Z$ reads as:
\begin{equation}\label{eq:subp-Z}
\min_{\Z} \;\;
    \frac{\gamma}{M}\|\X - \Z\|_{\mathrm{F}}^2 + \mathrm{Tr}\{ M^{-1}\Z\Z^\top \bdO^{t\top} \mathcal{L}_\Kronecker(\w^t)\bdO^t\}. \nonumber
\end{equation}
Since $\bdO^{t\top} \mathcal{L}_\Kronecker(\w^t)\bdO^t$ is positive semidefinite, the objective is convex, and the first-order optimality condition gives the closed form solution:
\begin{equation}\label{eq:update-Z}
    \Z^{t+1}=\big(\gamma\eye{Vn}+\bdO^{t\top} \mathcal{L}_\Kronecker(\w^t)\bdO^t\big)^{-1}\gamma\X .
\end{equation}
Equation~\eqref{eq:update-Z} acts as an infinite-impulse-response (IIR) graph filter~\cite{isufi2024graph} induced by the smoothness prior, which plays the role of a Tikhonov regularizer.

\spara{Update for $\w$.}
From \eqref{eq:aul}, the subproblem in $\w$ reads as:
\begin{equation}\label{eq:subp-w}
    \begin{aligned}
        \min_{\w\ge \mathbf{0}} &\;\; \mathrm{Tr}\{ M^{-1}\Z^{t+1}\Z^{t+1\top} \bdO^{t\top}\mathcal{L}_\Kronecker(\w)\,\bdO^t\} \\
        &\;\; + \frac{n\beta}{2} \|\mathcal{L}(\w) - \U^t \bm\Lambda^t \U^{t\top}\|_{\mathrm{F}}^2
        + \alpha\, \Psi(\w) .
    \end{aligned}
\end{equation}
The objective combines a convex smooth term, with a possibly nonsmooth, nonconvex sparsity-inducing penalty (SIP) $\Psi$. For simplicity, we adopt the concave smooth log-sum penalty $\Psi(\w)=\sum_i\log(\w_i+\epsilon)$~\cite{kumar2020unified}, which induces reweighted $\ell_1$ regularization~\cite{pmlr-v235-papamarkou24a}. Then, we minimize (\ref{eq:subp-w}) via an inexact minorization-maximization (MM) strategy~\cite{7547360}, performing one step of the block successive upper-bound minimization (BSUM) method~\cite{razaviyayn2013unified}. 

Specifically, let $g(\w)$ denote the smooth objective in~\eqref{eq:subp-w}, and let $\tau$ be a Lipschitz constant of $\nabla g$. Then, a strongly convex surrogate for $g(\w)$ satisfying all the conditions of BSUM is:
\begin{equation}\label{eq:surrogate-w}
     g(\w^t)+(\w-\w^t)^\top\nabla_{\w}g(\w^t) + \frac{1}{\tau}\|\w-\w^t\|^2.
\end{equation}
The surrogate in~\eqref{eq:surrogate-w} is a quadratic upper bound of the
objective in \eqref{eq:subp-w} around the current iterate $\w^t$. Its gradient with respect
to $\w$ is obtained by applying the adjoint operators
$\mathcal{L}_\Kronecker^\star$ and $\mathcal{L}^\star$ to the residuals
associated with the Kronecker-structured and combinatorial Laplacians,
respectively \cite{boyd2004convex}. Specifically, define
\begin{equation}
    f_1^t
    \coloneqq
    M^{-1}\bdO^t\Z^{t+1}\Z^{t+1\top}\bdO^{t\top},
\end{equation}
\begin{equation}\label{eq:f_2_surr}
    f_2^t
    \coloneqq
    n\beta\!\left[
    \mathcal{L}(\w^t)
    -\U^t\bm\Lambda^t\U^{t\top}
    \right],
\end{equation}
\begin{equation}
    f_3^t
    \coloneqq
\alpha \nabla_{\w}\Psi(\w^t),
\end{equation}
with $[\nabla_{\w}\Psi(\w^t)]_i=(\w_i^t+\epsilon)^{-1}$. Thus, minimizing the surrogate in~\eqref{eq:surrogate-w} over the nonnegative orthant yields the projected-gradient update
\begin{equation}\label{eq:w_update}
\w^{t+1}
=
\left[
\w^t-\frac{1}{\tau}
\left(
\mathcal{L}_\Kronecker^\star(f_1^t)
+
\mathcal{L}^\star(f_2^t) + f_3^t
\right)
\right]^+,
\end{equation}
where $[\cdot]^+$ denotes the projection onto $\mathbb{R}_+^{V(V-1)/2}$. To implement~\eqref{eq:w_update}, we require explicit expressions for the step-size constant $\tau$ and for the adjoint operators $\mathcal{L}_\Kronecker^\star$ and $\mathcal{L}^\star$. 
The stepsize is determined by the Lipschitz constant of the gradient of the smooth term $g(\mathbf{w})$ in the subproblem in $\mathbf{w}$. As shown in~\cite{kumar2020unified}, the Lipschitz constant for $\mathcal{L}^*\mathcal{L}$ is $2V$, so it follows from \eqref{eq:f_2_surr} that $\tau=2V\beta n$. 
The following lemmas characterize $\mathcal{L}_\Kronecker^\star$, consistently recovering $\mathcal{L}^\star$ as also derived in~\cite{kumar2020unified} for $n=1$:

\begin{lemma}[Adjoint operator]\label{def:AdjKLapl}
Let $\mathbf{Y}\in\mathbb{R}^{Vn\times Vn}$ be partitioned into
$n\times n$ blocks $\mathbf{Y}_{ij}$. The adjoint of
$\mathcal{L}_{\Kronecker}$, defined in~\cref{def:KLapl}, is the operator
$\mathcal{L}_{\Kronecker}^\star:\mathbb{R}^{Vn\times Vn}
\rightarrow\mathbb{R}^{V(V-1)/2}$ given by
\begin{equation}\label{eq:Lapl_star_K}
    \bigl[\mathcal{L}_{\Kronecker}^\star(\mathbf{Y})\bigr]_k
    =
    \operatorname{Tr}\!\left(
        \mathbf{Y}_{ii}+\mathbf{Y}_{jj}
        -\mathbf{Y}_{ij}-\mathbf{Y}_{ji}
    \right),
\end{equation}
where $k = i - j + \frac{(j-1)}{2}(2V - j)$ indexes the edge $(i,j)$.
\end{lemma}
\begin{proof}
See Appendix \ref{app:proof_lemma2}.
\end{proof}

\spara{Update for $\bdO$.}
From \eqref{eq:aul}, the subproblem in $\bdO$ reads as:
\begin{equation}\label{eq:subp-O}
\begin{aligned}
    \min_{\bdO}\quad
    & \operatorname{Tr}\!\left\{
        M^{-1}\bdO \Z^{t+1}\Z^{t+1\top}\bdO^\top
        \mathcal{L}_\Kronecker(\w^{t+1})
    \right\} \\
    & {}+ \frac{\rho}{2}
    \left\lVert \bdO-\myP^t+\myB^t \right\rVert_{\mathrm{F}}^2 .
\end{aligned}
\end{equation}
Subproblem~\eqref{eq:subp-O} is strongly convex in $\bdO$. Setting its gradient to zero yields
\begin{equation}\label{eq:closed_form_O}
    \bdO^{t+1}
    =
    \mathrm{vec}^{-1}\!\left[
        \left(2\bm{\Xi}+\rho\eye{V^2n^2}\right)^{-1}
        \rho\,\mathrm{vec}\!\left(\myP^t-\myB^t\right)
    \right],
\end{equation}
where
\[
\bm{\Xi}
=
M^{-1}\Z^{t+1}\Z^{t+1\top}
\otimes
\mathcal{L}_\Kronecker(\w^{t+1}).
\]
The inverse in~\eqref{eq:closed_form_O} is evaluated efficiently through the spectral decomposition described in Appendix~\ref{app:spectral}.

\spara{Update for $\myP$.} From \eqref{eq:aul}, the subproblem in $\myP$ reads as:
\begin{equation}\label{eq:subp-P}
    \min_{\myP \in \specialOprod{n}{V}}\; \frac{\rho}{2} \|\myP - \bdO^{t+1}+ \myB^t\|_{\mathrm{F}}^2.
\end{equation}
Let $\widetilde{\myP}^t=\bdO^{t+1}-\myB^t$. The $\myP$-update is obtained by projecting each diagonal block of $\widetilde{\myP}^t$ onto $\specialO{n}$ using the Kabsch algorithm~\cite{kabsch1976solution}:
\begin{equation}\label{eq:update-P}
    \myP_{ii}^{t+1}
    =
    \mathrm{Retr}\!\left(\widetilde{\myP}_{ii}^t\right),
    \qquad \forall\, i\in[V].
\end{equation}
Given the singular value decomposition $\widetilde{\myP}_{ii}=\mathbf{U}_{ii}\bm\Sigma_{ii}\mathbf{V}_{ii}^\top$, the retraction is $\mathrm{Retr}(\widetilde{\myP}_{ii})=\mathbf{U}_{ii}\widetilde{\bm\Sigma}_{ii}\mathbf{V}_{ii}^\top$ with $\widetilde{\bm\Sigma}_{ii}=\mathrm{diag}(1,\dots,1,\det(\mathbf{U}_{ii}\mathbf{V}_{ii}^\top))$, ensuring $\mathrm{Retr}(\widetilde{\myP}_{ii}^t)\in\specialO{n}$.

\spara{Update for $\U$.}  From \eqref{eq:aul}, the subproblem in $\U$ can be recast as an eigenvalue problem on $\stiefel{V}{V-1}$, i.e.,
\begin{equation}\label{eq:subp-U}
    \max_{\U \in \stiefel{V}{V-1}}\;\operatorname{Tr} \left\{\U^\top\mathcal{L}(\w^{t+1})\U\bm\Lambda^t \right\} ,
\end{equation}
whose solution comprises the $V-1$ eigenvectors associated with the nonzero eigenvalues of $\mathcal{L}(\w^{t+1})$~\cite{absil2009optimization}.

\spara{Update for $\bm\Lambda$.}
With $\w^{t+1}$ and $\U^{t+1}$ fixed, the $\bm\Lambda$-subproblem in~\eqref{eq:aul} reduces to an isotonic regression problem:
\begin{equation}\label{eq:IsoReg}
    \min_{c_1\leq\lambda_2\leq\cdots\leq\lambda_V\leq c_2}
    -n\sum_{i=1}^{V-1}\log(\lambda_{i+1})
    +\frac{n\beta}{2}\sum_{i=1}^{V-1}
    \left([\mathbf{M}]_{ii}-\lambda_{i+1}\right)^2,
\end{equation}
where $\mathbf{M}
=
\U^{t+1\top}\mathcal{L}(\w^{t+1})\U^{t+1}.$ The bounds $c_1$ and $c_2$ constrain the nonzero Laplacian eigenvalues away from degenerate values. Following~\cite{kumar2020unified}, we initialize the isotonic-regression procedure at the unconstrained minimizer, i.e.,
\begin{equation}\label{eq:LAMBDA_KKT}
    \lambda_{i+1}
    =
    \frac{1}{2}\left(
        [\mathbf{M}]_{ii}
        +
        \sqrt{[\mathbf{M}]_{ii}^2+\frac{4}{\beta}}
    \right),
    \qquad \forall\, i\in[V-1].
\end{equation}
The procedure in~\cite{kumar2020unified} then enforces the ordering and box constraints in~\eqref{eq:IsoReg} and converges in at most $V-1$ iterations.

\spara{Update for $\myB$.}
The scaled dual variable is updated by accumulating the residual of the
splitting constraint $\bdO=\myP$:
\begin{equation}\label{eq:update-B}
    \myB^{t+1}
    =
    \myB^t+\bdO^{t+1}-\myP^{t+1}.
\end{equation}

\begin{algorithm}[t]
\caption{Structured Connection Graph Learning (SCGL)}\label{alg:SCGL}
\begin{algorithmic}[1]
\REQUIRE $\X,\alpha, \beta, \w^{0}, \bdO^{0}, \myP^0=\bdO^0, \myB^0=\zeros, \U^{0}, \bm\Lambda^{0}, \Z^0 =\X$
\STATE $t \gets 0$
\WHILE{not converged}
    \STATE Update $\Z^{t+1}$ via~\eqref{eq:update-Z}
    \STATE Update $\w^{t+1}$ via~\eqref{eq:w_update}
    \STATE Update $\bdO^{t+1}$ via~\eqref{eq:closed_form_O}
    \STATE Update $\myP^{t+1}$ via~\eqref{eq:update-P}
    \STATE Update $\myB^{t+1}$ via~\eqref{eq:update-B}
    \STATE Update $\U^{t+1}$ via~\eqref{eq:subp-U}
    \STATE Update $\bm\Lambda^{t+1}$ via the isotonic algorithm~\cite{kumar2020unified} initialized at~\eqref{eq:LAMBDA_KKT}
    \STATE $t \gets t+1$
\ENDWHILE
\ENSURE $\Z^t,\w^{t}, \bdO^{t}$
\end{algorithmic}
\end{algorithm}

\Cref{alg:SCGL} summarizes the proposed Structured Connection Graph Learning (SCGL) procedure. The convergence conditions and the resulting asymptotic convergence to a stationary point of~\eqref{eq:final-prob} are established in the following theorem.

\begin{theorem}[Convergence of SCGL]\label{th:convergence}
Suppose that the spectral constraints are feasible. Then, the sequence
generated by \Cref{alg:SCGL} is bounded, and every one of its accumulation
points satisfies the KKT conditions of~\eqref{eq:final-prob}.
\end{theorem}
\begin{proof}
See Appendix~\ref{app:proofs-convergence}.
\end{proof}

In practice, the algorithm stops when all monitored primal-variable updates and the residual associated with the splitting constraint $\bdO=\myP$ fall below prescribed tolerances.

\spara{Computational complexity.}
The per-iteration cost of \cref{alg:SCGL} is dominated by the $\bdO$ update, that can be implemented with complexity
$O(V^3n^3)$. For $n=1$, SCGL reduces to the classical GSP setting and
recovers the $O(V^3)$ complexity of~\cite{kumar2020unified}. Compared with
the semidefinite-programming formulation in~\cite{8683709}, SCGL also relies
on a more compact parametrization: consistency represents the edge transports
through $V$ node reference frames, so $\bdO$ has $Vn^2$ parameters, rather
than the $V^2n^2$ parameters required to model all edge maps independently.

\subsection{Initialization Strategy}\label{subsec:init-strategy}

Since~\eqref{eq:final-prob} is nonconvex, the performance of the SCGL algorithm depends on the initialization. We therefore adopt the following procedure to obtain an initial point in a favorable region of the feasible set. To this aim, let $\kappa_0=\dim\ker(\Lapl)$ denote the nullity of the underlying combinatorial Laplacian. Under consistency,
$\dim\ker(\connectionL)=n\kappa_0$. Given $M$ noisy observations of $\x$, the population covariance is
\begin{equation}\label{eq:cov-noisy}
    \mathbb{E}\!\left[M^{-1}\X\X^\top\right]
    =
    \connectionL^\dagger+\sigma^2\eye{Vn}.
\end{equation}
Let $\mathbf{C}=M^{-1}\X\X^\top=\mathbf{W}\bm\Delta\mathbf{W}^\top$
be the eigendecomposition of the empirical covariance, with eigenvalues
sorted in ascending order. We estimate $\kappa_0$ using the Akaike
information criterion of~\cite{wax2003detection}, restricting the candidate
nullities of $\connectionL$ to multiples of $n$. Given the resulting estimate
$\hat{\kappa}_0$, the noise variance is estimated by averaging the smallest
$n\hat{\kappa}_0$ eigenvalues of $\mathbf{C}$. We then remove the isotropic
noise component by spectral thresholding,
\[
    [\widetilde{\bm\Delta}]_{ii}
    =
    \max\!\left\{[\bm\Delta]_{ii}-\hat{\sigma}^2,\,0\right\},
    \qquad
    \widetilde{\mathbf{C}}
    =
    \mathbf{W}\widetilde{\bm\Delta}\mathbf{W}^\top.
\]
The initial weights and node reference frames are then obtained as the
least-squares structured approximation of
$\widetilde{\mathbf{C}}^\dagger$:
\begin{equation}\label{eq:pseudoinvL}
    \min_{\substack{
        \w\in\reall_{\geq0}^{V(V-1)/2}\\
        \bdO\in\specialOprod{n}{V}
    }}
    \left\|
        \widetilde{\mathbf{C}}^\dagger
        -
        \bdO^\top\mathcal{L}_\Kronecker(\w)\bdO
    \right\|_{\mathrm{F}}^2.
\end{equation}
We solve~\eqref{eq:pseudoinvL} by Riemannian block-coordinate
descent~\cite{boumal2023introduction}. The remaining variables are
initialized through their corresponding update rules in~\cref{alg:SCGL},
with $\Z^0=\X$, and $\gamma=(2\hat{\sigma}^2)^{-1}$. In summary, SCGL is initialized by fitting a structured least-squares approximation of the pseudoinverse of the denoised empirical covariance.
\section{Numerical Results}\label{sec:results}

We evaluate SCGL on synthetic and real-world data, assessing both its recovery performance and its practical behavior.\footnote{Code available at \url{https://github.com/SPAICOM/consistent-connection-graph-learning}.}

\subsection{Impact of the Hyperparameters}

\begin{figure*}[t!]
    \centering
    \includegraphics[width=0.95\linewidth]{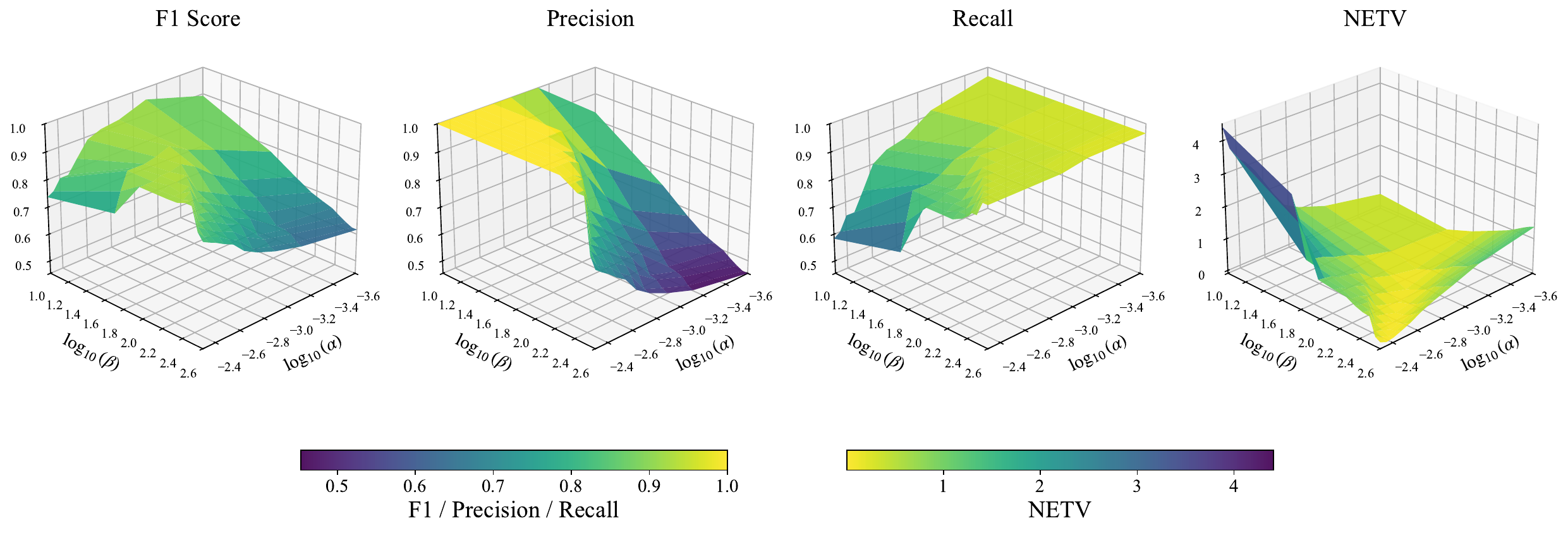}
    \caption{Ablation study on the hyperparameters $\alpha$ and $\beta$.}
    \label{fig:experiment_A}
\end{figure*}

We first study the behavior of SCGL under different hyperparameter settings, as the problem is inherently sensitive to the relative weighting of the sparsity and spectral penalties. We fix a planar $6\times 6$ lattice of $V=36$ vertices, generate a consistent CG over it with stalk dimension $n=3$, and draw $M=1080$ noiseless samples $\x \sim \mathcal{N}(\zeros, \connectionL^\dagger)$. We then test SCGL on $100$ combinations of $(\alpha,\beta)$ over the grid $[0.00025,0.005]\times[10,400]$.

Performance is assessed along two axes. Topological recovery is measured by the F1 score together with precision and recall. Geometric recovery is measured by the empirical total variation over a test set, used as a proxy because direct error evaluation on the learned local frames is ill-posed owing to the symmetry structure of the $\mathrm{SO}(n)$ group (Remark~\ref{remark:gauge_invariance}). The empirical total variation can instead be compared against its theoretical expectation,
\begin{equation}\label{eq:theor-val}
    \mathbb{E}[\mathrm{Tr}\{M^{-1}\X\X^\top\connectionL\}] = \mathrm{Tr}\{\connectionL^\dagger\connectionL\}=Vn-\dim\ker(\connectionL) ,
\end{equation}
with deviations from this value signaling incorrect inference of the geometric structure. Accordingly, the rightmost panel of~\cref{fig:experiment_A} reports the normalized empirical total variation,
\begin{equation}\label{eq:netv}
    \mathrm{NETV} = \frac{\big|\mathrm{Tr}\{\X\X^\top\widehat{\connectionL}\} - \mathbb{E}[\mathrm{Tr}\{\X\X^\top\connectionL\}]\big|}{\mathbb{E}[\mathrm{Tr}\{\X\X^\top\connectionL\}]} .
\end{equation}

\Cref{fig:experiment_A} reveals stable operating regions and a clear trade-off governed by the relative weighting of the sparsity penalty $\alpha$ and the spectral penalty $\beta$, mirroring the tension between recovering topology and preserving geometry. 
Increasing $\beta$ generally improves recall at the expense of precision, as it may introduce false-positive edges. This trade-off arises because enforcing consistency near the ground truth requires the inferred topology to contain, at least approximately, the ground-truth graph as a subgraph. Consequently, for $\beta/\alpha \gg 1$ the NETV approaches zero and nearly all true edges are recovered. Topological recovery can be sharpened further by favoring stronger sparsification while preserving accurate edge detection, driving the F1 score toward its optimum at the cost of exact geometric reconstruction. Guided by this trade-off, the experiments below use parameter settings for which the F1 score and the NETV are jointly near-optimal.

\begin{table*}[h!]
\centering
\caption{Average metrics ($\pm$ 1 SD) for random graph inference.}
\label{table:experiment_2}

\begin{subtable}[h]{0.48\textwidth}
\centering
\scalebox{.9}{
\begin{tabular}{|l|cccc|}
\hline
& & ER Graph & Random Geometric Graph & SBM \\
\hline
\hline
\multicolumn{1}{|l|}{\multirow{4}{*}{$\frac{M}{Vn}=1.5$}}
& SCOP & $0.40 \pm 0.06$ & $0.40 \pm0.04$ & $0.45\pm0.04$ \\
& SDP & $0.23 \pm 0.03$ & $0.26\pm0.05$ & $0.36 \pm 0.09$ \\
& SLGP & $0.40\pm0.09$ & $0.49\pm0.10$ & $0.39 \pm0.08$ \\
& SCGL & $\bf0.54 \pm 0.05$ & $\bf0.62\pm0.07$ & $\bf0.47\pm0.05$ \\
\hline
\multicolumn{1}{|l|}{\multirow{4}{*}{$\frac{M}{Vn}=5$}}
& SCOP & $0.51 \pm 0.05$ & $0.62\pm0.04$ & $0.68\pm0.03$ \\
& SDP & $0.23 \pm 0.03$ & $0.26\pm0.04$ & $0.39\pm0.05$ \\
& SLGP & $0.37\pm0.10$ & $0.59\pm0.07$ & $0.55\pm0.11$ \\
& SCGL & $\bf0.74 \pm 0.08$ & $\bf0.85\pm0.08$ & $\bf0.77\pm0.04$ \\
\hline
\multicolumn{1}{|l|}{\multirow{4}{*}{$\frac{M}{Vn}=15$}}
& SCOP & $0.57 \pm 0.05$ & $0.72\pm0.04$ & $0.75\pm0.02$ \\
& SDP & $0.23 \pm0.03$ & $0.26\pm0.04$ & $0.40\pm0.05$ \\
& SLGP & $0.39\pm0.12$ & $0.55\pm0.16$ & $0.54\pm0.06$ \\
& SCGL & $\bf0.79 \pm 0.07$ & $\bf0.94 \pm0.04$ & $\bf0.82\pm0.04$ \\
\hline
\end{tabular}}
\caption{F1 score ($\uparrow$) for topological recovery.}
\end{subtable}
\hfill
\begin{subtable}[h]{0.48\textwidth}
\centering
\scalebox{.9}{
\begin{tabular}{|l|cccc|}
\hline
& & ER Graph & Random Geometric Graph & SBM \\
\hline
\hline
\multicolumn{1}{|l|}{\multirow{4}{*}{$\frac{M}{Vn}=1.5$}}
& SCOP & $3.70 \pm 0.90$ & $4.70 \pm1.36$ & $3.23\pm0.96$ \\
& SDP & $0.60 \pm 0.55$ & $2.33\pm1.76$ & $0.78 \pm 0.17$ \\
& SLGP & $0.65\pm0.06$ & $0.30\pm0.20$ & $0.44 \pm0.07$ \\
& SCGL & $\bf0.45 \pm 0.16$ & $\bf0.08\pm0.06$ & $\bf0.62\pm0.21$ \\
\hline
\multicolumn{1}{|l|}{\multirow{4}{*}{$\frac{M}{Vn}=5$}}
& SCOP & $0.60 \pm 0.22$ & $0.76\pm0.26$ & $0.47\pm0.13$ \\
& SDP & $0.60 \pm 0.55$ & $2.43\pm1.86$ & $0.77\pm0.17$ \\
& SLGP & $0.64\pm0.07$ & $0.33\pm0.21$ & $0.44 \pm0.08$ \\
& SCGL & $\bf0.03 \pm 0.02$ & $\bf0.03\pm0.03$ & $\bf0.05\pm0.02$ \\
\hline
\multicolumn{1}{|l|}{\multirow{4}{*}{$\frac{M}{Vn}=15$}}
& SCOP & $0.24 \pm 0.12$ & $0.28\pm0.10$ & $0.16\pm0.07$ \\
& SDP & $0.59 \pm0.54$ & $2.40\pm1.83$ & $0.77\pm0.16$ \\
& SLGP & $0.65\pm0.06$ & $0.35\pm0.21$ & $0.43 \pm0.09$ \\
& SCGL & $\bf0.02 \pm 0.01$ & $\bf0.02 \pm0.02$ & $\bf0.02\pm0.01$ \\
\hline
\end{tabular}}
\caption{Normalized empirical total variation ($\downarrow$) for geometric recovery.}
\end{subtable}
\end{table*}

\begin{figure*}[t]
    \centering
    \includegraphics[width=0.85\linewidth]{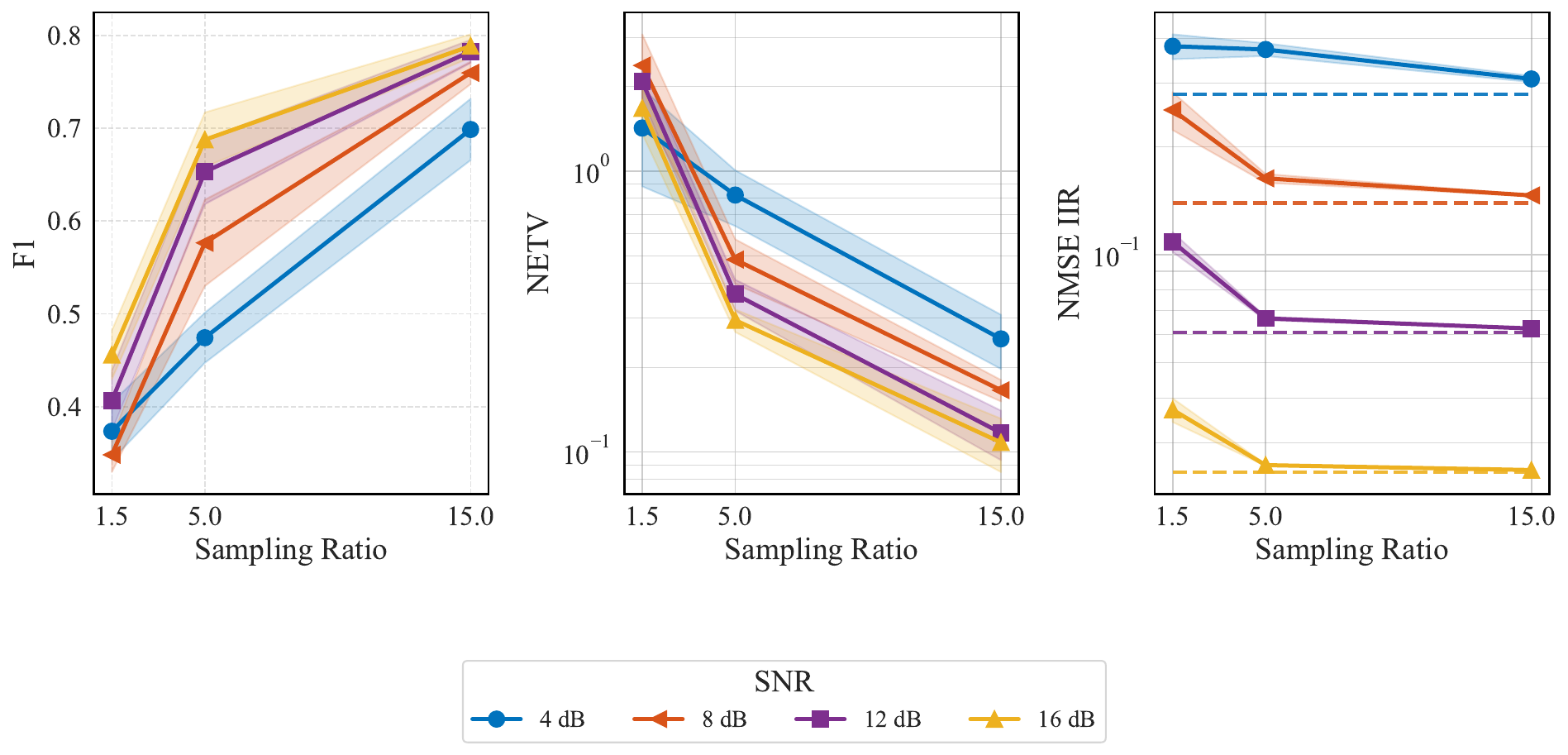}
    \caption{Inference in noisy settings across signal-to-noise ratios: (left) F1 score for topology detection; (center) normalized empirical total variation as a proxy for geometric recovery; (right) NMSE between the noiseless signal and its denoised estimate obtained by the IIR filter parametrized with the inferred connection Laplacian. Dashed lines mark the NMSE attained by the IIR filter parametrized with the ground-truth connection Laplacian and noise variance.}
    \label{fig:experiment_C}
\end{figure*}

\subsection{Random Graph Inference: Noiseless Setting}

We consider graphs with $V = 30$ nodes and stalk dimension $n = 2$ drawn from three random models:
(i) \textbf{Erdős--Rényi (ER)} graphs with edge probability $p = 1.1\,\tfrac{\log V}{V}$ and weights sampled from $\mathrm{Unif}(0.2, 3)$;
(ii) \textbf{random geometric graphs}, obtained by sampling points uniformly in the unit cube and assigning weights through the Gaussian kernel $\mathbf{w}_{ij}=\exp\{-\tfrac{\|\mathbf{x}_{i} -\mathbf{x}_j\|^2}{2\zeta^2}\}\,\mathbf{1}(\|\mathbf{x}_{i}-\mathbf{x}_j\|_2^2\leq 0.375)$ with $\zeta=0.5$;
and (iii) \textbf{stochastic block models (SBM)} with $3$ communities, intra-cluster probability $p_{\mathrm{in}} = 0.7$, inter-cluster probability $p_{\mathrm{out}}=0.05$, community assignment $[0.4,0.3,0.3]$, and unit edge weights. For each graph we build a consistent CG by sampling node bases from $\mathrm{SO}(2)$, then draw $M$ noiseless samples $\x \sim \mathcal{N}(\zeros, \connectionL^\dagger)$ under three data-availability regimes: \emph{low} ($M/Vn = 1.5$), \emph{medium} ($M/Vn = 5$), and \emph{high} ($M/Vn = 15$). Each graph and regime is repeated over $20$ independent simulations.

We compare SCGL against three baselines:
(i) \emph{structured covariance pseudo-inversion} (SCOP), the initialization procedure of~\cref{subsec:init-strategy};
(ii) \emph{smooth learning via SDP} (SDP), the semidefinite-programming approach
of~\cite{8683709} restricted to Laplacians with identity-scaled diagonal
blocks---this enforces smoothness but optimizes over an enlarged convex cone that strictly contains the set of connection Laplacians, and therefore does not respect the CG Riemannian geometry; and (iii) \emph{smooth learning with geometric priors} (SLGP), our earlier method~\cite{10942997}, which solves a minimum-total-variation problem with known edge-set cardinality and edge maps constrained to $\mathrm{O}(n)$. For SCGL we fix $\beta=60$, $\alpha=0.0025$; remaining hyperparameters are chosen by cross-validation, no post-processing is applied to the learned Laplacians, and the number of connected components is supplied as prior knowledge.

We assess the F1 score and the NETV in~\eqref{eq:netv} for topological and geometric recovery, respectively; the results are collected in~\cref{table:experiment_2}. As we can see from \cref{table:experiment_2}, SCGL consistently improves upon SCOP: the latter supplies a sound starting point on which SCGL builds steady gains in both topology and geometry. SCGL is also the method most responsive to refinement of the empirical covariance estimate, as expected from its derivation from the prescribed signal model, which in turn exposes the limitations of SDP and SLGP in learning CGs with prescribed spectral structure. 
The generative model markedly affects performance, with ER and SBM failing for opposite reasons: ER lacks the structural scaffold needed to guide alignment, while SBM's strong community structure pits intra-cluster coherence against inter-cluster alignment. Both effects peak in the low-data regime, and although more data helps, the topology–geometry trade-off stays fragile for ER and SBM. SCGL is strongest on random geometric graphs—in both accuracy and robustness—effectively neutralizing this trade-off.

\subsection{Random Graph Inference: Noisy Setting}

We next assess robustness to measurement noise. We fix a random geometric graph,
construct a consistent CG as above, and generate training signals from the full
probabilistic model in~\eqref{eq:marginal-x}, varying the noise variance to
obtain different signal-to-noise ratios (SNRs). For each sampling ratio--SNR
pair, results are averaged over $10$ realizations; the F1 score and NETV are
evaluated on noiseless test signals $\X_{\mathrm{test}}$. As shown in
\Cref{fig:experiment_C} (left and center), both metrics improve with the
sampling ratio and SNR. Their dependence on the sampling ratio also follows
the trends observed in the noiseless setting, indicating robustness to noisy
training observations.

We further assess the denoising capability of the learned Laplacian. The test
signals are corrupted at the same SNR used during training, yielding
$\widetilde{\X}_{\mathrm{test}}$, and denoised with the IIR filter
$\mathbf{H}_{\widehat{\connectionL}}
=
\gamma\left(\gamma\eye{Vn}+\widehat{\connectionL}\right)^{-1}$. We report
$\mathrm{NMSE}(\mathbf{H}_{\widehat{\connectionL}}
(\widetilde{\X}_{\mathrm{test}}),\X_{\mathrm{test}})$ in
\Cref{fig:experiment_C} (right) as a function of the sampling ratio for
different SNR values. The dashed horizontal line denotes the reference NMSE
obtained by applying the same filter with the ground-truth connection
Laplacian and noise variance.
As we can notice from \Cref{fig:experiment_C} (right), the inferred Laplacian is a strong denoiser: its NMSE nearly matches the theoretical bound even at low SNR, providing further empirical support for the approach.

\subsection{Inference Under Model Mismatch}

In this section, we evaluate the framework under model mismatch in a geometric scenario relevant to spatial statistics. We build a sphere in $\reall^3$ as a $k$-NN graph ($k=4$) over a Fibonacci lattice~\cite{stanley1975fibonacci} with $V=100$ nodes, and equip it with a connection Laplacian approximated by Vector Diffusion Maps (VDM)~\cite{singer2012vector}. This CG is inherently \emph{non}-consistent: the intrinsic curvature of the sphere prevents the edge maps from factorizing, so the spectrum of the connection Laplacian is untied from that of the underlying graph and its kernel is trivial.

We introduce a second source of mismatch through a \emph{random signal model} derived from the Kraichnan model for random vector fields~\cite{kraichnan1970diffusion}. Writing a point as the concatenation of its coordinates $\mathbf{v}=(v_1,v_2,v_3)$, we define a zero-mean Gaussian process on $\reall^3$ with radial covariance $K(\mathbf{v},\mathbf{v}')=\exp\{-\tfrac{\|\mathbf{v}-\mathbf{v}'\|_2^2}{2\zeta^2}\}$, $\zeta=0.5$. Drawing $N=1000$ samples from the spectral density of this covariance and setting the field-variance parameter to $\sigma^2=1$, the signal $\mathbf{f}\in\reall^{3V}$ at a point $\mathbf{v}$ along coordinate $\iota\in\{1,2,3\}$ is synthesized spectrally as
\begin{equation}
    \mathbf{f}_\iota(\mathbf{v})= \sqrt{\tfrac{\sigma^2}{N}}\sum_{j=1}^N\mathbf{p}_\iota(\mathbf{k}_j)\big[Z_{1,j}\cos(\mathbf{k}_j^\top\mathbf{v})+Z_{2,j}\sin(\mathbf{k}_j^\top\mathbf{v})\big] ,
\end{equation}
where $Z_{l,j}\sim\mathcal{N}(0,\sigma^2)$ for $l\in\{1,2\}$ and the spectral samples are projected as $\mathbf{p}_\iota(\mathbf{k}_j)=(\eye{3}-\tfrac{1}{\langle\mathbf{k}_j,\mathbf{k}_j\rangle}\mathbf{k}_j\mathbf{k}_j^\top)\mathbf{e}_\iota$, with $\{\mathbf{e}_{\iota}\}_{\iota=1,2,3}$ the canonical basis of $\reall^3$. The model is implemented with the \texttt{GSTools} library~\cite{muller2022gstools}. To embed the vector-field signals $\mathbf{f}\in\reall^{3V}$ into CG signals $\tilde{\mathbf{f}}\in\reall^{2V}$, we adopt the sampler of~\cite{battiloro2024tangent}: the local signal $\tilde{\mathbf{f}}(i)$ is the projection of $\mathbf{f}(\mathbf{v})$ onto the tangent space at node $i$ defined by the VDM local frames (see~\cite[Sec.~IV-B]{battiloro2024tangent}).

We benchmark SCGL against the baselines on a compression task. From a training set of $2000$ signals we learn connection Laplacians and extract their eigenvectors as an orthogonal basis for the CG signals; we then apply this basis to a test set of $1000$ signals from the same model, performing top-$k$ nonlinear compression~\cite{vetterli2014foundations}, and record the NMSE between the compressed signals $\hat{\mathbf{f}}$ and the originals across sparsity levels. For reference we also include the eigenbasis of the VDM connection Laplacian and that of a trivial-connection graph ($\mathbf{O}_i=\eye{n}$ for all $i\in[V]$), i.e., a classical graph inferred by minimum total variation (KRON).
\begin{figure}[t]
    \centering
    \includegraphics[width=0.95\linewidth]{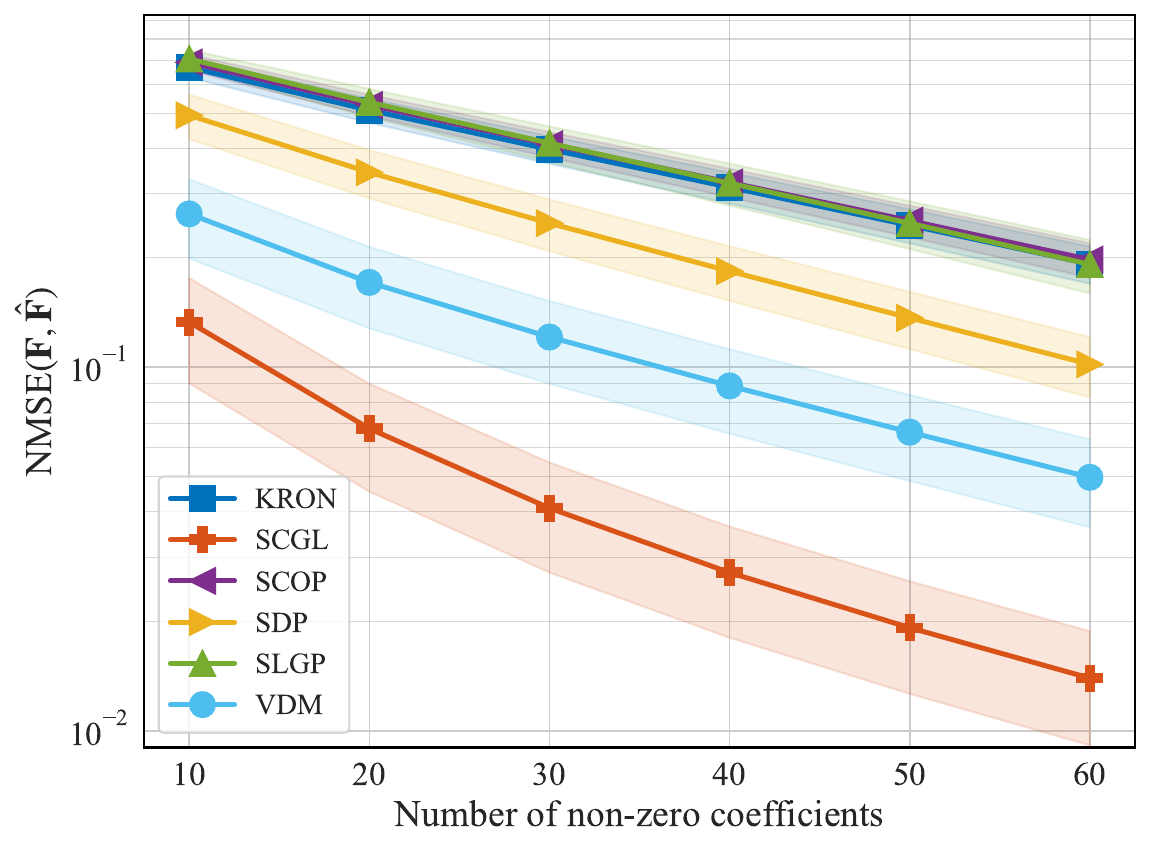}
    \caption{Compression–reconstruction trade-off for random vector fields over a discrete sphere, using the eigenbases of various estimated Laplacian operators.}
    \label{fig:experiment_D}
\end{figure}
The results in~\cref{fig:experiment_D} expose the limitations of classical graph learning and demonstrate the advantage of sheaf-based learning in improving the reconstruction–compression trade-off as sparsity increases. 
The eigenbasis of the connection Laplacian estimated by SCGL achieves the best compression–reconstruction performance among the data-driven Laplacians. 
SCGL improves also with respect to compression over VDM eigenbasis, the ideal operator for discrete vector fields~\cite{battiloro2024tangent}. 
Thus, even under model mismatch SCGL is effective for sheaf signal processing on such domains. 
Crucially, VDM and SCGL differ in a broader perspective: VDM builds its connection Laplacian from the point-cloud coordinates, whereas SCGL relies solely on the observed signals.
VDM's transports are estimated from local neighborhoods and recover the true tangent-bundle connection only asymptotically in the sampling density~\cite{singer2012vector}; under coarse  sampling these estimates might be noisy, and the resulting basis may carry that error into the downstream task. 
SCGL, being signal-adaptive, sidesteps this dependence and can therefore improve over VDM precisely in the sparsely-sampled, despite violating the underlying geometric assumption.
% SCGL is therefore applicable when the latent geometry is unknown, opening a promising avenue toward transferring flat-bundle geometry to more general classes of cellular sheaves.

\subsection{Learning Graphs of RotatedMNIST digits}
We evaluate our algorithm's ability to learn a graph over \texttt{RotatedMNIST}~\cite{larochelle2007empirical}. This is a challenging task, motivated by the prominence of this benchmark in rotational-equivariant learning~\cite{worrall2017harmonic,weiler2018learning}: each image is equipped with an $\mathrm{SO}(2)$ rotation matrix that alters the orientation of its pixels. This property makes the dataset especially well suited to our framework, as we can assign each image to a node of a consistent connection graph with unknown node reference frames. Building on similar experiments~\cite{kalofolias2016learn,wasserman2024graph}, we target the learning of graphs with good clustering properties, and show how the spectral constraint generalizes from connected graphs to $k$-connected-component graphs that separate RotatedMNIST images in an unsupervised way.

To apply our algorithm to the rotated images, we extract from each image a data matrix $\mathbf{F}_i\in\mathbb{R}^{2\times T}$ whose features are equivariant to in-plane rotation, so that rotating the image transforms its features by the corresponding rotation. We realize this feature engineering with steerable filters~\cite{93808}, implemented as Gaussian--Laguerre filters of angular order $h$ and radial index $j$, with $\rho=r/w$ and $w=0.35$ a bandwidth parameter:
$$
\psi_{h,j}(\rho,\theta)\propto\rho^{h}L_{j}^{(h)}(\rho^{2})e^{-\rho^{2}/2}e^{-ih\theta}
$$
where $L_j^{(h)}(x)=\sum_{k=0}^{j}(-1)^k\binom{j+h}{j-k}\frac{x^k}{k!}$.
Cross-correlating the image $\mathcal{I}$ with a filter bank of angular order $h=0,\ldots,H=3$ and radial index $j=1,\ldots,J=8$ gives responses $\hat{c}_{h,j}$. We form $T=HJ^2$ equivariant features given by the stack the conjugate products  $\hat{c}_{h,j_1}\overline{\hat{c}_{h-1,j_2}},h=1,\ldots,H+1, j_1=1,\ldots,J, j_2=1,\ldots,J$ into $\mathbf{F}_i\in\mathbb{R}^{2\times T}$~\cite{worrall2017harmonic}.

Following prior works aiming at discriminating digits ``1'' from ``2'', we sample $10$ instances per class, rotate each by a uniformly random angle, embed them as connection graph signals through the feature pipeline above, and run SCGL with a spectral prior encoding the number of classes.
\begin{figure}[t]
    \centering

    \begin{subfigure}{\linewidth}
        \centering
        \includegraphics[width=\linewidth]{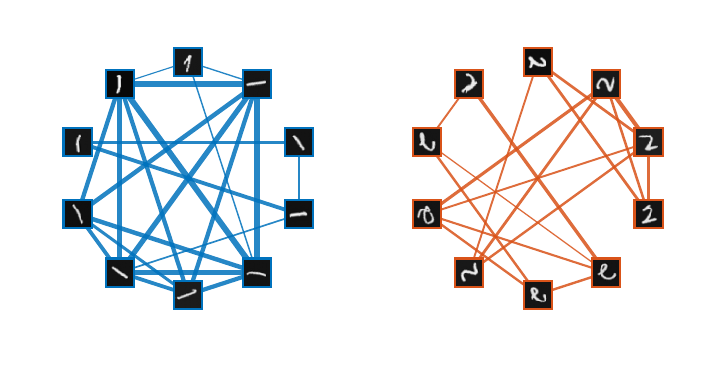}
        \caption{Graph learned over instances of RotatedMNIST, with the prior of 2 connected components.}
        \label{fig:rotatedMNIST}
    \end{subfigure}

    \vspace{0.5em}

    \begin{subfigure}{\linewidth}
        \centering
        \includegraphics[width=\linewidth]{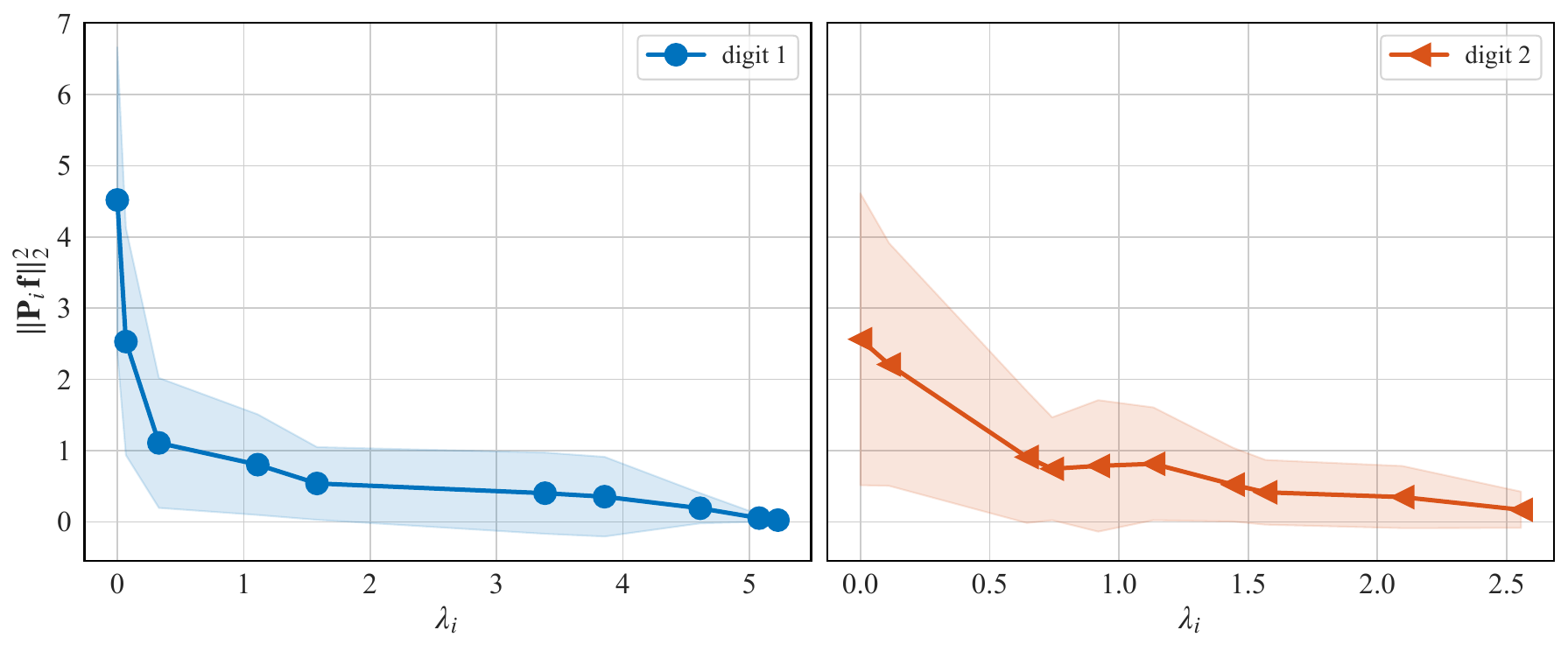}
        \caption{Spectral profile of the MNIST signals over the learned graph.}
        \label{fig:rotatedMNIST_spectrum}
    \end{subfigure}

    \caption{Results on the RotatedMNIST dataset.}
    \label{fig:rotatedMNIST_results}
\end{figure}

\Cref{fig:rotatedMNIST} shows the graphs learned by the SCGL method. It illustrates that combining local geometric structure with a global spectral constraint yields robust representations in unsupervised settings: the learned graph exactly recovers the imposed prior of two connected components, associated with the images of the digits “1” and “2,” respectively. 

Furthermore, \Cref{fig:rotatedMNIST_spectrum}  shows the spectral behavior of the training signals, obtained through the connection Fourier transform (CFT), across the frequencies of the learned graph. Since the connection Laplacian reads as $\mathbb{L}=\mathrm{blkdiag}(\mathbb{L}_1,\mathbb{L}_2)$ with $\mathbb{L}_k=\mathbf{U}_k(\mathbf{\Lambda}_k\otimes\mathbf{I}_2)\mathbf{U}_k^\top,k=1,2,$ each graph frequency spans a two-dimensional eigenspace with orthogonal projector $\mathbf{P}_i=\sum_{j=1}^{2}\mathbf{u}_{i,j}\mathbf{u}_{i,j}^\top$. Thus, for every signal $\mathbf{f}$, the CFT evaluates the behavior of $\|\mathbf{P}_i\mathbf{f}\|_2^2$ across eigenspaces. Interestingly, the CFT profiles in \Cref{fig:rotatedMNIST_spectrum} are markedly smoother for the “1” cluster than for the “2” cluster. 
This difference can be attributed to the greater average dissimilarity among images of the digit ``2,'' which results in a lower concentration of energy at low graph frequencies. These results confirm that SCGL captures rotational variability while recovering the structural, rotation-invariant properties of the signals of interest.
\section{Conclusions}\label{sec:conclusion}

In this paper, we introduced SCGL, a framework for learning consistent connection graphs and their associated connection Laplacians directly from observed signals. 
Connection graphs form a structured class of network sheaves in which edge transports encode geometric relations between local signal spaces; their consistency property provides a tractable link between this geometry and the spectrum of an underlying combinatorial graph. 
Exploiting this structure, SCGL jointly infers the network topology, the node reference frames that define the edge transports, and denoised signal estimates. 
The proposed formulation further transfers spectral constraints from structured graph learning to connection graphs, enabling explicit control of the inferred operator while preserving its geometric consistency. 
Numerical experiments on synthetic and real-world data demonstrate the resulting benefits for both network identification and signal processing.

Several directions remain open. 
First, the probabilistic model could be extended to the broader class of stationary Gaussian graph signals. 
A further direction is to relax the consistency assumption and develop efficient mechanisms for learning or approximating non-consistent connection Laplacians. 
The framework could also be reformulated in a task-driven setting, coupling connection-graph learning with downstream objectives. 
Finally, SCGL may provide a principled basis for designing neural estimators of consistent connection graphs within sheaf-based deep learning architectures.
\balance
\appendix
\subsection{Spectral implementation of the SOC update for $\bdO$}
\label{app:spectral}

A direct evaluation of~\eqref{eq:closed_form_O} requires inverting the
$V^2n^2\times V^2n^2$ matrix
$2\bm{\Xi}+\rho\mathbf{I}$, with complexity
$\mathcal{O}(V^6n^6)$. Exploiting the Kronecker structure of
$\bm{\Xi}$ yields a much cheaper implementation. Let
\begin{align}
    M^{-1}\Z\Z^\top &= \U_{\Z}\bm\Lambda_{\Z}\U_{\Z}^\top, \nonumber\\
    \mathcal{L}_{\Kronecker}(\w)
    &= \U_{\mathcal{L}}\bm\Lambda_{\mathcal{L}}\U_{\mathcal{L}}^\top,
    \label{eq:spectral_decompositions}
\end{align}
and define
\[
\mathbf{R}^t
=
\U_{\mathcal{L}}^\top
(\myP^t-\myB^t)
\U_{\Z},
\qquad
\bm\Phi
=
2\,\bm\lambda_{\mathcal{L}}
\bm\lambda_{\Z}^{\top}
+
\rho\,\mathbf{1}\mathbf{1}^{\top},
\]
where $\bm\lambda_{\mathcal{L}}$ and $\bm\lambda_{\Z}$ collect the eigenvalues
of $\bm\Lambda_{\mathcal{L}}$ and $\bm\Lambda_{\Z}$, respectively.
Using the identity
$\mathrm{vec}(\mathbf{A}\mathbf{X}\mathbf{B})
=
(\mathbf{B}^{\top}\otimes\mathbf{A})
\mathrm{vec}(\mathbf{X})$,
the linear system in~\eqref{eq:closed_form_O} diagonalizes in the spectral
domain, yielding
\begin{equation}\label{eq:spectral_SOC}
    \bdO^{t+1}
    =
    \rho\,
    \U_{\mathcal{L}}
    \left(
        \mathbf{R}^t
        \oslash
        \bm\Phi
    \right)
    \U_{\Z}^{\top},
\end{equation}
where $\oslash$ denotes elementwise division.

The dominant cost is given by the two eigendecompositions in~\eqref{eq:spectral_decompositions},
each requiring $\mathcal{O}(V^3n^3)$ operations; the remaining spectral
filtering and matrix multiplications have lower complexity. Thus, the overall
cost of the update is $\mathcal{O}(V^3n^3)$.

\vspace{-.3cm}
\subsection{Proof of Lemma~\ref{def:AdjKLapl}}
\label{app:proof_lemma2}
We use the Frobenius inner products on
$\mathbb{R}^{Vn\times Vn}$ and $\mathbb{R}^{V(V-1)/2}$.
\begin{equation}
        \begin{aligned}
            &\langle \mathcal{L}_{\mathbb{K}}(\mathbf{w}), \mathbf{Y} \rangle = \\&- \sum_{i>j} \mathrm{Tr}\{\mathbf{w}_{i+d_j}\mathbf{Y}_{ij}\} - \sum_{i<j} \mathrm{Tr}\{\mathbf{w}_{j+d_i}\mathbf{Y}_{ij}\} \\ & + \sum_i \left(\mathrm{Tr}\{ \sum_{j<i} \mathbf{w}_{i+d_j} \mathbf{Y}_{ii}\} + \mathrm{Tr}\{ \sum_{j>i} \mathbf{w}_{j+d_i} \mathbf{Y}_{ii}\} \right) =\\ 
            &= \sum_{i>j} \mathrm{Tr}\{\mathbf{w}_{i+d_j}\left(\mathbf{Y}_{ii} + \mathbf{Y}_{jj} - \mathbf{Y}_{ij} - \mathbf{Y}_{ji}\right)\} = \\
            & =\sum_{i>j} \mathbf{w}_{i+d_j}\mathrm{Tr}{\left(\mathbf{Y}_{ii} + \mathbf{Y}_{jj} - \mathbf{Y}_{ij} - \mathbf{Y}_{ji}\right)}= \\
            &= \langle \mathbf{w}, \mathcal{L}_{\mathbb{K}}^\star(\mathbf{Y}) \rangle \,.
        \end{aligned}
    \end{equation}
where $d_j = -j+\frac{j-1}{2}(2V-j)$ and similarly for $d_i$.
Let $k=i+d_j$, then \Cref{eq:Lapl_star_K} follows.

\vspace{-.2cm}
\subsection{Proof of Theorem~\ref{th:convergence}}
\label{app:proofs-convergence}

We apply the BSUM convergence result of~\cite{razaviyayn2013unified} and the
two-block extension used in the proof of Theorem~7
of~\cite{kumar2020unified}.

\spara{Boundedness:} The spectral constraints bound $\bm\Lambda$, while $\U$ and $\myP$ belong to
the compact manifolds $\mathrm{St}(V,V{-}1)$ and $\mathrm{SO}(n)^V$.
Coercivity of the objective in $(\Z,\w,\bm\Lambda)$ bounds $\Z$ and $\w$;
the closed-form update~\eqref{eq:closed_form_O} then yields a bounded
$\bdO$. Hence, the iterates are bounded and admit accumulation points.

\spara{Convergence:} We apply the result of~\cite{razaviyayn2013unified}, together
with the two-block extension used in the proof of Theorem~7
of~\cite{kumar2020unified}. We verify its conditions by showing that each
block update minimizes either the exact block objective or a valid upper-bound
surrogate that is tight and first-order consistent at the current iterate.
In particular, the $\Z$-update in~\eqref{eq:update-Z} and the
$\bdO$-update in~\eqref{eq:closed_form_O} are closed-form solutions of
strongly convex quadratic subproblems. The $\w$-update in~\eqref{eq:w_update}
minimizes a strongly convex quadratic majorizer followed by projection onto
the nonnegative orthant, while the $\bm\Lambda$-update solves the ordered
spectral subproblem~\eqref{eq:IsoReg}, initialized at~\eqref{eq:LAMBDA_KKT}. The $\U$-update in~\eqref{eq:subp-U} and the $\myP$-update
in~\eqref{eq:update-P} may be non-unique because of the symmetries of
$\mathrm{St}(V,V{-}1)$ and $\mathrm{SO}(n)^V$, respectively. The cited
extension in~\cite{kumar2020unified} permits two non-unique updates when the corresponding blocks are mutually decoupled. This holds here: \eqref{eq:subp-U} depends only on $\mathcal{L}(\w^{t+1})$, whereas \eqref{eq:update-P} depends only on $\bdO^{t+1}-\myB^t$. Therefore, every accumulation point satisfies the
block-wise first-order optimality conditions associated with the augmented
Lagrangian in~\eqref{eq:aul}.

\spara{KKT conditions:} At an accumulation point, these conditions yield stationarity of the
unconstrained $\Z$-block, the projected first-order condition for $\w$, the
KKT conditions of~\eqref{eq:IsoReg} for $\bm\Lambda$, and Riemannian
stationarity of $\U$ and $\myP$. The dual update~\eqref{eq:update-B} enforces
$\bdO^\star=\myP^\star$. Finally,  the linear independence constraint qualification (LICQ) holds for the spectral constraints by
\cite[Lemma~17]{kumar2020unified}, while the Stiefel and special orthogonal
constraints are regular because $\mathrm{St}(V,V{-}1)$ and
$\mathrm{SO}(n)^V$ are smooth embedded manifolds. The resulting
block-wise optimality and feasibility conditions are therefore the KKT
conditions of~\eqref{eq:final-prob}.
\bibliographystyle{IEEEtran}
\bibliography{settings/bibliography}

\end{document}